\documentclass[american,aps,prl,reprint,longbibliography,superscriptaddress]{revtex4-2}
\pdfoutput=1
\usepackage[T1]{fontenc}
\usepackage[utf8]{inputenc}
\usepackage{color}
\usepackage{babel}
\usepackage{amsmath}
\usepackage{amsthm}
\usepackage{amssymb}
\usepackage{physics}
\usepackage{graphicx}
\usepackage[unicode=true,pdfusetitle,
 bookmarks=true,bookmarksnumbered=false,bookmarksopen=false,
 breaklinks=false,pdfborder={0 0 0},pdfborderstyle={},backref=false,colorlinks=true]
 {hyperref}
\hypersetup{
 allcolors=magenta}

\makeatletter
\theoremstyle{plain}
\newtheorem{thm}{\protect\theoremname}
\theoremstyle{plain}

\theoremstyle{plain}

\ifx\proof\undefined
\newenvironment{proof}[1][\protect\proofname]{\par
	\normalfont\topsep6\p@\@plus6\p@\relax
	\trivlist
	\itemindent\parindent
	\item[\hskip\labelsep\scshape #1]\ignorespaces
}{%
	\endtrivlist\@endpefalse
}
\providecommand{\proofname}{Proof}
\fi

\usepackage{times}
\usepackage{txfonts}
\usepackage{braket}
\usepackage{colortbl}

\makeatother

\providecommand{\propositionname}{Proposition}
\providecommand{\theoremname}{Theorem}
\providecommand{\lemmaname}{Lemma}

\begin{document}
\title{Catalytic and asymptotic equivalence for quantum entanglement}

\author{Ray Ganardi}
\email{r.ganardi@cent.uw.edu.pl}
\affiliation{Centre for Quantum Optical Technologies, Centre of New Technologies,
University of Warsaw, Banacha 2c, 02-097 Warsaw, Poland}
\author{Tulja Varun Kondra}
\affiliation{Centre for Quantum Optical Technologies, Centre of New Technologies,
University of Warsaw, Banacha 2c, 02-097 Warsaw, Poland}
\author{Alexander Streltsov}
\affiliation{Institute of Fundamental Technological Research, Polish Academy of Sciences, Pawińskiego 5B, 02-106 Warsaw, Poland}

\begin{abstract}
    Entanglement is a fundamental resource in quantum information processing, yet understanding its manipulation and transformation remains a challenge.  Many tasks rely on highly entangled pure states, but obtaining such states is often challenging due to the presence of noise. Typically, entanglement manipulation procedures involving asymptotically many copies of a state are considered to overcome this problem. These procedures allow for distilling highly entangled pure states from noisy states, which enables a wide range of applications, such as quantum teleportation and quantum cryptography. When it comes to manipulating entangled quantum systems on a single copy level, using entangled states as catalysts can significantly broaden the range of achievable transformations. Similar to the concept of catalysis in chemistry, the entangled catalyst is returned unchanged at the end of the state manipulation procedure. Our results demonstrate that despite the apparent conceptual differences between the asymptotic and catalytic settings, they are actually strongly connected and fully equivalent for all distillable states. Our methods rely on the analysis of many-copy entanglement manipulation procedures which may establish correlations between different copies. As an important consequence, we demonstrate that using an entangled catalyst cannot enhance the asymptotic singlet distillation rate of a distillable quantum state. Our findings provide a comprehensive understanding of the capabilities and limitations of both catalytic and asymptotic state transformations of entangled states, and highlight the importance of correlations in these processes.
\end{abstract}

\maketitle

\textbf{\emph{Introduction.}} Entanglement is a key feature of quantum mechanics, and it plays a vital role in many areas of quantum information science. Being a strong form of correlations between quantum systems, entanglement enables a wide range of applications and protocols that have the potential to revolutionize information processing and communication~\cite{PhysRevLett.70.1895,PhysRevLett.67.661}. The study of entanglement and its properties has led to significant advancements in our understanding of quantum mechanics, and it has provided insights into how to manipulate and harness its power for practical applications~\cite{HorodeckiRevModPhys.81.865}.

To understand the pivotal role of entanglement as a resource in quantum information processing, we can consider the distant lab paradigm~\cite{BennettPhysRevLett.76.722,BennettPhysRevA.54.3824}.  This scenario involves two parties, Alice and Bob, who are located in different quantum laboratories and can exchange classical messages to communicate with each other. In this setting, entangled states shared between Alice and Bob become a valuable resource, allowing them to perform tasks that would otherwise be impossible~\cite{HorodeckiRevModPhys.81.865}. 

One of the most significant applications of entanglement is in the field of quantum communication, including quantum teleportation~\cite{PhysRevLett.70.1895} and quantum cryptography~\cite{PhysRevLett.67.661}. These tasks typically rely on singlets, which are pure highly entangled states of two qubits. However, in practice, Alice and Bob may only have access to noisy states. In order to use noisy states for singlet-based protocols, they can employ entanglement distillation~\cite{BennettPhysRevLett.76.722,BennettPhysRevA.53.2046}, which is a special case of asymptotic state transformations. In this process, $n$ copies of an initial state are transformed approximately into $rn$ copies of the final state, where $r$ is the transformation rate. Quantum states which can be distilled into singlets at a nonzero rate are called distillable. There exist noisy entangled states which cannot be distilled into singlets, a phenomenon known as bound entanglement~\cite{PhysRevLett.80.5239}.

Another way how Alice and Bob can gain access to singlets from noisy states is to use entanglement catalysis. In this process, an auxiliary entangled state, known as a catalyst, is employed to aid in the transformation of one entangled state to another without altering the catalyst itself~\cite{JonathanPhysRevLett.83.3566}. Recent work~\cite{PhysRevLett.127.150503,Lipka-Bartosik2102.11846,Rubboli2111.13356,datta2022entanglement} extended this idea to approximate catalysis, where the transformation can be achieved with a certain degree of inaccuracy. This concept has proven to be instrumental in advancing our understanding of catalytic entanglement manipulation and its potential applications~\cite{catalysis_review}.

At first glance, catalytic and asymptotic transformations may seem like distinct concepts, but recent research has uncovered a strong connection between them. Initial evidence for a connection between these concepts was presented in~\cite{PhysRevA.72.024306,DuanPhysRevA.71.042319}, and subsequent work has made significant progress in this direction, particularly through the use of approximate catalysis~\cite{PhysRevLett.127.150503,Rubboli2111.13356}. Furthermore, it has been shown that in quantum thermodynamics, catalysis and many-copy transformations with a unit rate are fully equivalent~\cite{shiraishi2020quantum,wilming2020entropy}. Given the shared features between quantum entanglement and thermodynamics~\cite{PhysRevLett.89.240403,Popescu2006,Brandao2008,Lami2023}, it is plausible that a similar equivalence may exist in entanglement theory.

In this article, we resolve this question by considering catalytic and asymptotic protocols which can establish a non-vanishing amount of correlations. This provides a more flexible approach for studying catalysis and asymptotic transformations and their applications in quantum information processing. In this setting, we prove that for distillable states, catalysis and asymptotic transformations with unit rate are fully equivalent notions of entangled state manipulation. We discuss several applications of our results, including the crucial finding that the addition of a catalyst cannot increase the distillable entanglement of a noisy distillable state.

\textbf{\emph{Asymptotic entanglement manipulations and catalysis.}} As previously discussed, asymptotic transformations are a powerful tool for understanding the structure and manipulation of quantum entanglement. For instance, consider two bipartite pure states $\ket{\psi}$ and $\ket{\phi}$. The objective is to use local operations and classical communication (LOCC) to transform $n$ copies of $\ket{\psi}$ into $m$ copies of $\ket{\phi}$, allowing for an error margin that vanishes in the limit of large $n$. The maximal ratio $m/n$ defines the transformation rate. This framework is particularly useful if the target is the singlet state $\ket{\psi^-} = (\ket{01}- \ket{10})/\sqrt{2}$, in which case the optimal rate is known as distillable entanglement~\cite{BennettPhysRevLett.76.722,BennettPhysRevA.53.2046,PlenioMeasures}. It coincides with the entanglement entropy $E(\ket{\psi}) = S(\psi^A)$ of the initial state, and $S(\rho) = -\mathrm{Tr[\rho \log_2 \rho]}$ is the von Neumann entropy~\cite{BennettPhysRevA.53.2046}. In a similar way, it is possible to define transformation rates for noisy states, we refer to the Supplemental Material~\footnote{See Supplemental Material for technical explanations and more details, which contains the additional references~\cite{Eggeling_2001,959270,Devetak2005,Matthews_2008,PhysRevA.65.012107,Palazuelos2022,PhysRevA.75.012305,PhysRevLett.102.170503,Christandl_2004,Alicki_2004,RevModPhys.91.025001}} for more details. A state is called asymptotically reducible onto another state if the transformation can be achieved with a rate at least one~\cite{PhysRevA.63.012307}.
This reflects the intuition that if $\ket{\psi}$ is reducible onto $\ket{\phi}$, then $\ket{\psi}$ is at least as valuable as $\ket{\phi}$ for any application that allows for asymptotic transformations.

Entanglement catalysis is a phenomenon where an entangled catalyst is used to facilitate the single-copy transformation of an entangled state into another without changing the state of the catalyst~\cite{JonathanPhysRevLett.83.3566,EisertPhysRevLett.85.437}. Given an entangled state $\ket{\psi}$ and a target state $\ket{\phi}$, the aim is to find a catalytic state $|\eta\rangle$ such that the transformation $\ket{\psi} \otimes \ket{\eta} \rightarrow \ket{\phi} \otimes \ket{\eta}$ is possible by LOCC. The catalyst is particularly useful if it enables the transformation of $\ket{\psi}$ into $\ket{\phi}$ which is not possible without the catalyst. Recently, the notion of catalysis has been extended to approximate catalysis in~\cite{PhysRevLett.127.150503,Lipka-Bartosik2102.11846,Rubboli2111.13356,datta2022entanglement}, which allows for some degree of inaccuracy in the catalytic transformation. The notion of approximate catalysis provides a more realistic model for practical implementations of catalytic entanglement manipulation and enables a broader range of applications~\cite{catalysis_review}. It has been demonstrated in~\cite{PhysRevLett.127.150503} that transformations between bipartite pure states in this scenario are fully determined by the entanglement entropy of the corresponding states. Therefore, for bipartite pure states approximate catalysis is fully equivalent to reducibility. Catalytic phenomena have been extensively studied not only in the context of entanglement, but also in other areas of quantum physics, such as quantum thermodynamics~\cite{Brandao3275,PhysRevX.8.041051,shiraishi2020quantum,BoesPhysRevLett.122.210402,wilming2020entropy}, where they are essential for understanding and manipulating quantum systems subject to constraints imposed by energy conservation.

As has been shown in~\cite{PhysRevLett.127.150503}, there is a close connection between asymptotic state transformations and catalysis. More precisely, if a state $\rho$ is asymptotically reducible to another state $\sigma$, then a transformation from $\rho$ into $\sigma$ can also be achieved on the single-copy level with approximate catalysis~\cite{PhysRevLett.127.150503}. However, it has remained a crucial open question whether the converse is also true, i.e., whether catalysis and asymptotic reducibility are fully equivalent notions for entangled state transformations. In this article, we introduce the frameworks of marginal asymptotic transformations and correlated catalysis, which allows us to resolve this question and establish the equivalence between catalysis and reducibility for all distillable quantum states.

\begin{figure*} 
\includegraphics[width=1\textwidth]{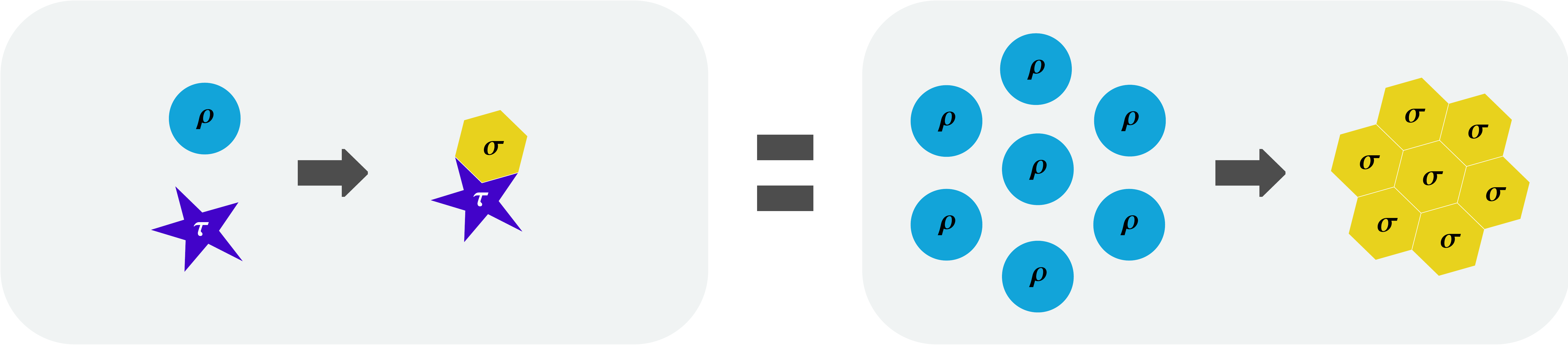}

\caption{\label{fig:Main} Equivalence between catalysis and reducibility. The left part of the figure shows a state $\rho$ being converted into another state $\sigma$ through a catalytic transformation by using a catalyst in the state $\tau$. The right part of the figure shows that $\rho$ is asymptotically reducible into $\sigma$. Our results demonstrate that the two processes are equivalent for any pair of distillable states, assuming that both procedures can establish correlations.}
\end{figure*}

\textbf{\emph{Correlated catalysis and marginal reducibility.}} In the context of entanglement catalysis, an important generalization is to consider \emph{correlated catalysis}, where the catalyst is allowed to have non-vanishing correlations with the system throughout the transformation process. This means that the system and the catalyst can remain correlated in the final state. More precisely, we say that $\rho$ can be converted into $\sigma$ via correlated catalysis if for any error margin $\varepsilon > 0$ there is an LOCC protocol $\Lambda$ and a catalyst state $\tau$ such that
\begin{equation} \label{eq:CorrelatedCatalysis}
\begin{aligned}\mu^{SC} & =\Lambda(\rho^{S}\otimes\tau^{C}),\\
||\mu^{S}-\sigma^{S}||_{1} & <\varepsilon,\,\,\,\mu^{C}=\tau^{C}.
\end{aligned}
\end{equation}
Here, $S$ denotes a possibly multipartite system, $C$ denotes the catalyst, and $||M||_1 = \mathrm{Tr}\sqrt{M^\dagger M}$ is the trace norm. In other words, the state $\mu^{SC}$ is obtained by applying an LOCC protocol $\Lambda$ to the state $\rho^{S}\otimes\tau^{C}$, such that the marginal on $C$ is preserved and the resulting state on $S$ can be made arbitrarily close to the target state $\sigma$. Previous studies in quantum thermodynamics have explored the significance of correlations for catalytic state transformations, revealing that the presence of correlations between the system and catalyst can increase the transformation power of the procedure~\cite{Wilminge19060241,BoesPhysRevLett.122.210402,wilming2020entropy,shiraishi2020quantum}.

We now introduce the notion of \emph{marginal reducibility}.
We say that $\rho$ can be reduced into $\sigma$ in the marginals if for any arbitrarily small error margin, there exists an LOCC protocol which can transform $n$ copies of $\rho$ into a state with approximately $m$ marginals, and each marginal being close to the desired state $\sigma$. Specifically, we require for any $\varepsilon, \delta > 0$ there exists an LOCC protocol $\Lambda$ and integers $m \leq n$ such that the following conditions hold for all $i \leq m$:
\begin{align} \label{eq:CorrelatedReducibility}
\Lambda\left(\rho^{\otimes n}\right) & = \mu_{m},\nonumber \\
\left\Vert \mu_{m}^{(i)}-\sigma\right\Vert _{1} & <\varepsilon,\\
\frac{m}{n}+\delta & >1. \nonumber
\end{align}
Here, $\mu_m$ is a state on $m$ subsystems, each shared by Alice and Bob, and $\mu_{m}^{(i)}$ is the reduced state of $\mu_m$ on $i$-th subsystem. Marginal asymptotic transformations have been previously studied in continuous variable systems in~\cite{Ferrari2023}.

It is worth to discuss the difference between marginal reducibility and the notion of reducibility introduced in~\cite{PhysRevA.63.012307}. The latter is more stringent as it requires that the final state $\mu_m$ is close to $m$ copies of $\sigma$ as a whole. However, for many quantum information processing tasks that rely on pure states $\ket{\phi}$, such as singlets in the bipartite case or GHZ states in the multipartite case, small perturbations of the state do not significantly affect its usefulness. In other words, the state $\mu_{\varepsilon}=(1-\varepsilon)\ket{\phi}\!\bra{\phi}+\varepsilon\openone/d$ is also useful for small enough $\varepsilon$. For marginal reducibility, it suffices that $\rho^{\otimes n}$ can be approximately converted into $\mu_\varepsilon^{\otimes n}$ for any $\varepsilon > 0$, which as we have argued above is enough for many tasks based on pure states. Therefore, we suggest that the framework of marginal reducibility is particularly suitable when one aims to produce pure entangled states of high quality that are intended to be used independently.

In the remainder of this article we will focus on the relationship between correlated catalysis and marginal reducibility. Unless otherwise specified, we will refer to these concepts simply as catalysis and reducibility, respectively.

\smallskip

\textbf{\emph{Catalysis-reducibility equivalence.}} As previously noted, there have been indications that catalysis and reducibility are interchangeable concepts for entangled state transformations~\cite{PhysRevA.72.024306,PhysRevLett.127.150503,Rubboli2111.13356,DuanPhysRevA.71.042319}.
The key contribution of this article is to establish this equivalence for any pair of distillable states, using the notions of catalysis and reducibility which include correlations, as outlined above in this article. We recall that a distillable state is a quantum state which can be converted into singlets at nonzero rate in the asymptotic limit.
\begin{thm} \label{thm:Main}
For any pair of distillable states $\rho$ and $\sigma$ reducibility and catalysis are fully equivalent.
\end{thm}

\noindent We present a brief overview of the techniques employed for proving the theorem, more details can be found in the Supplemental Material. Firstly, we demonstrate that if the state $\rho$ can be reduced to the state $\sigma$, then it is possible to achieve a catalytic transformation from $\rho$ to $\sigma$, using techniques similar to those presented in prior works~\cite{shiraishi2020quantum,PhysRevLett.127.150503,Lipka-Bartosik2102.11846,Rubboli2111.13356}. Subsequently, we establish the converse by explicitly constructing a reduction protocol that utilizes a catalytic conversion protocol from a distillable state $\rho$ into $\sigma$. This involves several technical steps that are described in detail in the Supplemental Material. By combining these two results, we conclusively demonstrate the full equivalence of reducibility and catalysis for any pair of distillable states $\rho$ and $\sigma$, see also Fig.~\ref{fig:Main}.

Since all entangled two-qubit states are distillable~\cite{HorodeckiPhysRevLett.78.574}, Theorem~\ref{thm:Main} implies that catalysis and reducibility are fully equivalent for all two-qubit states.
For states beyond two qubits, Theorem~\ref{thm:Main} also applies if the target state $\sigma$ is not distillable. Moreover, catalysis is generally at least as powerful as reducibility. With this in mind, Theorem~\ref{thm:Main} leaves open the possibility that there exist bound entangled states $\rho$ that cannot be reduced to some state $\sigma$, yet a catalytic conversion from $\rho$ to $\sigma$ is possible. Thus, if catalysis and reducibility are not equally powerful on all quantum states, catalysis must show an advantage on some bound entangled initial states. This underscores the importance of investigating the relationship between these concepts in the general case, as it can provide insights into the nature of bound entanglement and the power of entanglement catalysis. Additionally, our findings can have implications for quantum information processing tasks where bound entangled states are known to play a significant role~\cite{PhysRevLett.82.1056,PhysRevLett.86.2681,Vertesi2014}.

Going one step further, we investigate the role of catalysis for asymptotic transformation rates. Our findings reveal that the addition of a catalyst does not alter the asymptotic rate of transformation from a distillable state $\rho$ into another state $\sigma$, again under the assumption that correlations can be established in the procedure. An important application of this result pertains to the scenario where the target state is a singlet $\ket{\psi^-}$. In this context, our analysis reveals that the correlations, which are typically established in the catalytic and asymptotic procedures considered earlier, vanish. This property allows us to explore the features of distillable entanglement when a catalyst is incorporated into the transformation, bringing us to the second main result of this article. 

\begin{thm} \label{thm:DistillableEntanglement}
    Catalysis cannot increase the distillable entanglement of a distillable state.
\end{thm}

\noindent The proof of the theorem combines the previously mentioned results on asymptotic transformation rates with the additional finding that correlations usually established in the involved procedures disappear if the target state is pure. We refer to the Supplemental Material for the proof and more details. Recalling that all entangled two-qubit states can be distilled into singlets~\cite{HorodeckiPhysRevLett.78.574}, it follows that Theorem~\ref{thm:DistillableEntanglement} applies to all two-qubit states. In general, our results leave open the possibility that bound entangled states could be activated into singlets through catalysis. 

\begin{figure} 
\includegraphics[width=0.66\columnwidth]{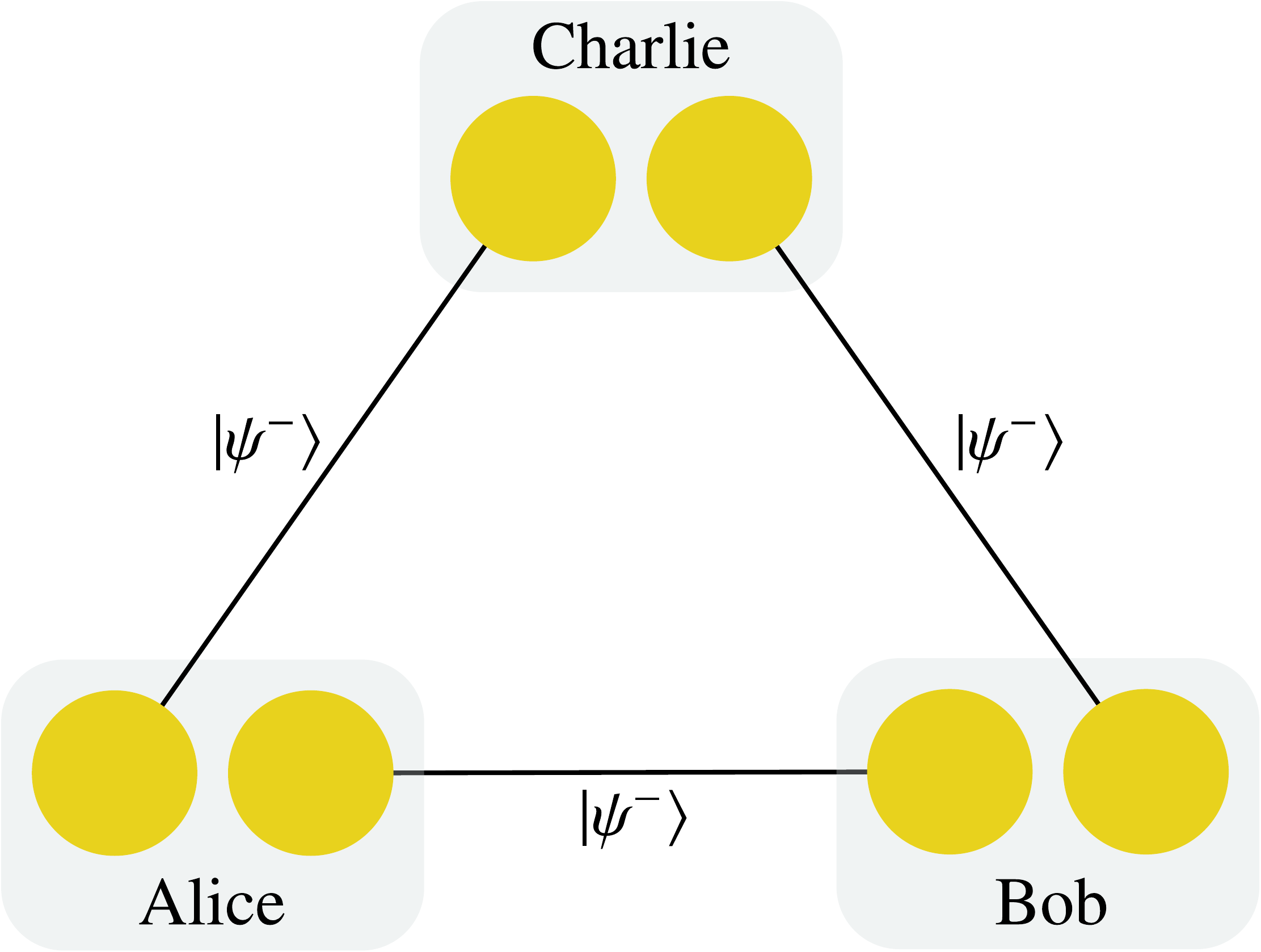}

\caption{\label{fig:Multipartite}Distillation of multipartite entanglement. We call a multipartite state distillable if it can be asymptotically converted into states comprising a singlet between each of the parties. The figure shows the desired final state for $3$ parties. The multipartite distillable entanglement is the maximal transformation rate into states of this form. With these definitions, Theorems~\ref{thm:Main} and~\ref{thm:DistillableEntanglement} extend to the multipartite setting.}
\end{figure}

Our results have implications also beyond the scope of bipartite systems. It is worth noting that Theorem~\ref{thm:Main} can be generalized to the multipartite scenario.
To this end, we consider multipartite distillable states, which are those (pure and mixed) multipartite states that can be distilled into singlets between each pair of parties with some nonzero rate in the asymptotic limit, see also Fig.~\ref{fig:Multipartite}. This includes all genuinely multipartite entangled pure states, i.e., pure states which are entangled across any bipartition~\cite{SmolinPhysRevA.72.052317,qsm,GUHNE20091}.
With this in mind, we can extend Theorem~\ref{thm:Main} to state that for any pair of multipartite distillable states, reducibility and catalysis are fully equivalent.
Furthermore, Theorem~\ref{thm:DistillableEntanglement} is also applicable to this scenario, indicating that the addition of a catalyst cannot enhance the multipartite distillable entanglement of any multipartite distillable state, we refer to the Supplemental Material for more details.
The results obtained in the multipartite setting are in line with those in the bipartite setting and imply that if catalysis offers any benefit over reducibility, it can only be observed when the initial state is not distillable. 

The limitation to distillable states in Theorem~\ref{thm:Main} can be overcome by allowing the borrowing of a pure state, that is, considering transformations from $\rho \otimes \psi$ to $\sigma \otimes \psi$ with some entangled pure state $\ket{\psi}$. In this case, the state $\rho \otimes \psi$ is distillable, leading to the equivalence of catalysis and reducibility for any $\rho$ and $\sigma$. Interestingly, this applies even if the borrowed state $\ket{\psi}$ has arbitrarily little entanglement. Similarly, we can extend Theorem~\ref{thm:DistillableEntanglement} to state that catalysis cannot increase the distillable entanglement of the state $\rho \otimes \psi$, where $\rho$ does not need to be distillable. 

The methods discussed in this article are applicable to a broader framework of entanglement theory, where the free states are the states that remain positive under partial transposition (PPT states), and the free operations are all quantum transformations whose Choi state is PPT~\cite{Rains_1999,Rains_1999a}. Within this framework, we establish the complete equivalence between reducibility and catalysis, i.e., Theorem~\ref{thm:Main} extends to all quantum states. We also show that Theorem~\ref{thm:DistillableEntanglement} applies to all states in this setting, i.e., catalysis cannot enhance the distillation rate in general. We refer to the Supplemental Material for more details. Intriguingly, this also implies that in the LOCC setting Theorem~\ref{thm:DistillableEntanglement} applies even if the initial state $\rho$ is PPT (see also~\cite{Lami2023b} for an independent proof). Thus, the only class of states for which Theorem~\ref{thm:DistillableEntanglement} is not known to hold in the LOCC setting is the class of NPT bound entangled states, the existence of which has not been confirmed~\cite{PhysRevA.61.062312,5508622}. 

There has been growing interest in other extensions of LOCC operations, such as non-entangling operations~\cite{Brandao2008, Lami2023} and dually non-entangling operations~\cite{Chitambar_2020, lami2023c}. Non-entangling operations represent the most general transformations in bipartite (or multipartite) settings that do not generate entanglement from separable states. It has been established that reversible asymptotic manipulation of entangled states is impossible even within this extended framework~\cite{Lami2023}. Furthermore, it has been demonstrated that maximally entangled mixed states for a fixed spectrum do not always exist under for LOCC or non-entangling operations~\cite{PhysRevLett.133.050202}. Dually non-entangling operations lie between LOCC and non-entangling operations, with the defining property of preserving both the sets of separable states and separable measurements~\cite{Chitambar_2020, lami2023c}.  Notably, the proof of Theorem~\ref{thm:Main} can be directly extended to cover both non-entangling and dually non-entangling operations, provided the initial state $\rho$ remains distillable under LOCC. This extends also to the multipartite setting, if we assume that the initial state can be distilled into singlets between each pair of the parties at a nonzero rate via LOCC.

These findings offer a better understanding of the relationship between entanglement catalysis and many-copy transformations, and can have practical implications for the exploitation of entanglement in quantum information processing tasks.

\smallskip

\textbf{\emph{Conclusions.}} In conclusion, our work establishes the complete equivalence between reducibility and catalysis for any pair of distillable states, which extends and confirms previous indications that these concepts are interchangeable for entangled state transformations. Furthermore, we have demonstrated that the addition of a catalyst does not alter the rate of asymptotic transformations between distillable states.

Our results shed new light on the nature of entanglement catalysis and entanglement-based protocols. The full equivalence between catalysis and reducibility for distillable states provides a clearer understanding of the limitations and capabilities of these tasks. We emphasize that our results assume that correlations can be established in the transformation procedures involved. This suggests that taking correlations into account can provide a more complete and accurate understanding of quantum information processing tasks which make use of catalysis. The methods developed in this article can guide the design of new protocols, where catalysis and correlations play a significant role. On the other hand, correlations disappear naturally for transformations into pure target states. This allows us to conclude that the addition of a catalyst cannot increase the asymptotic singlet distillation rate, provided that the initial state has non-zero distillable entanglement to begin with. 

The manipulation of entanglement in the multipartite setting is a complex and challenging problem~\cite{StreltsovPhysRevLett.125.080502}, and further research in this direction is required to fully understand and effectively utilize the power of multipartite entanglement in quantum information processing. Our findings are particularly relevant in this context, as they demonstrate the full equivalence of catalysis and reducibility for transformations between multipartite distillable states. These findings have significant implications for understanding the role of catalysis in communication protocols that rely on multipartite entangled states, such as quantum secret sharing~\cite{PhysRevLett.83.648,PhysRevA.59.1829}.

Furthermore, our work opens up new avenues for research into the relationship between reducibility and catalysis in the general case, where the assumption of distillability cannot be made. Investigating this relationship can help us better understand the nature of bound entanglement and unlock the full potential of entanglement and general quantum resources in quantum information processing. 

\emph{Note added.} In an independent work~\cite{Lami2023b} it has been shown with different methods that by using entanglement catalysis it is not possible to distill singlets from bound entangled states having positive partial transpose.

We thank Ludovico Lami for insightful comments on our manuscript. This work was supported by the ``Quantum Optical Technologies'' project, carried out within the International Research Agendas programme of the Foundation for Polish Science co-financed by the European Union under the European Regional Development Fund, and the National Science Centre Poland (Grant No. 2022/46/E/ST2/00115) and within the QuantERA II Programme (Grant No. 2021/03/Y/ST2/00178, acronym ExTRaQT) that has received funding from the European Union's Horizon 2020 research and innovation programme under Grant Agreement No. 101017733.

\bibliography{literature}

\begin{thebibliography}{59}%
\makeatletter
\providecommand \@ifxundefined [1]{%
 \@ifx{#1\undefined}
}%
\providecommand \@ifnum [1]{%
 \ifnum #1\expandafter \@firstoftwo
 \else \expandafter \@secondoftwo
 \fi
}%
\providecommand \@ifx [1]{%
 \ifx #1\expandafter \@firstoftwo
 \else \expandafter \@secondoftwo
 \fi
}%
\providecommand \natexlab [1]{#1}%
\providecommand \enquote  [1]{``#1''}%
\providecommand \bibnamefont  [1]{#1}%
\providecommand \bibfnamefont [1]{#1}%
\providecommand \citenamefont [1]{#1}%
\providecommand \href@noop [0]{\@secondoftwo}%
\providecommand \href [0]{\begingroup \@sanitize@url \@href}%
\providecommand \@href[1]{\@@startlink{#1}\@@href}%
\providecommand \@@href[1]{\endgroup#1\@@endlink}%
\providecommand \@sanitize@url [0]{\catcode `\\12\catcode `\$12\catcode
  `\&12\catcode `\#12\catcode `\^12\catcode `\_12\catcode `\%12\relax}%
\providecommand \@@startlink[1]{}%
\providecommand \@@endlink[0]{}%
\providecommand \url  [0]{\begingroup\@sanitize@url \@url }%
\providecommand \@url [1]{\endgroup\@href {#1}{\urlprefix }}%
\providecommand \urlprefix  [0]{URL }%
\providecommand \Eprint [0]{\href }%
\providecommand \doibase [0]{https://doi.org/}%
\providecommand \selectlanguage [0]{\@gobble}%
\providecommand \bibinfo  [0]{\@secondoftwo}%
\providecommand \bibfield  [0]{\@secondoftwo}%
\providecommand \translation [1]{[#1]}%
\providecommand \BibitemOpen [0]{}%
\providecommand \bibitemStop [0]{}%
\providecommand \bibitemNoStop [0]{.\EOS\space}%
\providecommand \EOS [0]{\spacefactor3000\relax}%
\providecommand \BibitemShut  [1]{\csname bibitem#1\endcsname}%
\let\auto@bib@innerbib\@empty
\bibitem [{\citenamefont {Bennett}\ \emph {et~al.}(1993)\citenamefont
  {Bennett}, \citenamefont {Brassard}, \citenamefont {Cr\'epeau}, \citenamefont
  {Jozsa}, \citenamefont {Peres},\ and\ \citenamefont
  {Wootters}}]{PhysRevLett.70.1895}%
  \BibitemOpen
  \bibfield  {author} {\bibinfo {author} {\bibfnamefont {C.~H.}\ \bibnamefont
  {Bennett}}, \bibinfo {author} {\bibfnamefont {G.}~\bibnamefont {Brassard}},
  \bibinfo {author} {\bibfnamefont {C.}~\bibnamefont {Cr\'epeau}}, \bibinfo
  {author} {\bibfnamefont {R.}~\bibnamefont {Jozsa}}, \bibinfo {author}
  {\bibfnamefont {A.}~\bibnamefont {Peres}},\ and\ \bibinfo {author}
  {\bibfnamefont {W.~K.}\ \bibnamefont {Wootters}},\ }\bibfield  {title}
  {\bibinfo {title} {{Teleporting an unknown quantum state via dual classical
  and Einstein-Podolsky-Rosen channels}},\ }\href
  {https://doi.org/10.1103/PhysRevLett.70.1895} {\bibfield  {journal} {\bibinfo
   {journal} {Phys. Rev. Lett.}\ }\textbf {\bibinfo {volume} {70}},\ \bibinfo
  {pages} {1895} (\bibinfo {year} {1993})}\BibitemShut {NoStop}%
\bibitem [{\citenamefont {Ekert}(1991)}]{PhysRevLett.67.661}%
  \BibitemOpen
  \bibfield  {author} {\bibinfo {author} {\bibfnamefont {A.~K.}\ \bibnamefont
  {Ekert}},\ }\bibfield  {title} {\bibinfo {title} {{Quantum cryptography based
  on Bell's theorem}},\ }\href {https://doi.org/10.1103/PhysRevLett.67.661}
  {\bibfield  {journal} {\bibinfo  {journal} {Phys. Rev. Lett.}\ }\textbf
  {\bibinfo {volume} {67}},\ \bibinfo {pages} {661} (\bibinfo {year}
  {1991})}\BibitemShut {NoStop}%
\bibitem [{\citenamefont {Horodecki}\ \emph {et~al.}(2009)\citenamefont
  {Horodecki}, \citenamefont {Horodecki}, \citenamefont {Horodecki},\ and\
  \citenamefont {Horodecki}}]{HorodeckiRevModPhys.81.865}%
  \BibitemOpen
  \bibfield  {author} {\bibinfo {author} {\bibfnamefont {R.}~\bibnamefont
  {Horodecki}}, \bibinfo {author} {\bibfnamefont {P.}~\bibnamefont
  {Horodecki}}, \bibinfo {author} {\bibfnamefont {M.}~\bibnamefont
  {Horodecki}},\ and\ \bibinfo {author} {\bibfnamefont {K.}~\bibnamefont
  {Horodecki}},\ }\bibfield  {title} {\bibinfo {title} {Quantum entanglement},\
  }\href {https://doi.org/10.1103/RevModPhys.81.865} {\bibfield  {journal}
  {\bibinfo  {journal} {Rev. Mod. Phys.}\ }\textbf {\bibinfo {volume} {81}},\
  \bibinfo {pages} {865} (\bibinfo {year} {2009})}\BibitemShut {NoStop}%
\bibitem [{\citenamefont {Bennett}\ \emph
  {et~al.}(1996{\natexlab{a}})\citenamefont {Bennett}, \citenamefont
  {Brassard}, \citenamefont {Popescu}, \citenamefont {Schumacher},
  \citenamefont {Smolin},\ and\ \citenamefont
  {Wootters}}]{BennettPhysRevLett.76.722}%
  \BibitemOpen
  \bibfield  {author} {\bibinfo {author} {\bibfnamefont {C.~H.}\ \bibnamefont
  {Bennett}}, \bibinfo {author} {\bibfnamefont {G.}~\bibnamefont {Brassard}},
  \bibinfo {author} {\bibfnamefont {S.}~\bibnamefont {Popescu}}, \bibinfo
  {author} {\bibfnamefont {B.}~\bibnamefont {Schumacher}}, \bibinfo {author}
  {\bibfnamefont {J.~A.}\ \bibnamefont {Smolin}},\ and\ \bibinfo {author}
  {\bibfnamefont {W.~K.}\ \bibnamefont {Wootters}},\ }\bibfield  {title}
  {\bibinfo {title} {{Purification of Noisy Entanglement and Faithful
  Teleportation via Noisy Channels}},\ }\href
  {https://doi.org/10.1103/PhysRevLett.76.722} {\bibfield  {journal} {\bibinfo
  {journal} {Phys. Rev. Lett.}\ }\textbf {\bibinfo {volume} {76}},\ \bibinfo
  {pages} {722} (\bibinfo {year} {1996}{\natexlab{a}})}\BibitemShut {NoStop}%
\bibitem [{\citenamefont {Bennett}\ \emph
  {et~al.}(1996{\natexlab{b}})\citenamefont {Bennett}, \citenamefont
  {DiVincenzo}, \citenamefont {Smolin},\ and\ \citenamefont
  {Wootters}}]{BennettPhysRevA.54.3824}%
  \BibitemOpen
  \bibfield  {author} {\bibinfo {author} {\bibfnamefont {C.~H.}\ \bibnamefont
  {Bennett}}, \bibinfo {author} {\bibfnamefont {D.~P.}\ \bibnamefont
  {DiVincenzo}}, \bibinfo {author} {\bibfnamefont {J.~A.}\ \bibnamefont
  {Smolin}},\ and\ \bibinfo {author} {\bibfnamefont {W.~K.}\ \bibnamefont
  {Wootters}},\ }\bibfield  {title} {\bibinfo {title} {Mixed-state entanglement
  and quantum error correction},\ }\href
  {https://doi.org/10.1103/PhysRevA.54.3824} {\bibfield  {journal} {\bibinfo
  {journal} {Phys. Rev. A}\ }\textbf {\bibinfo {volume} {54}},\ \bibinfo
  {pages} {3824} (\bibinfo {year} {1996}{\natexlab{b}})}\BibitemShut {NoStop}%
\bibitem [{\citenamefont {Bennett}\ \emph
  {et~al.}(1996{\natexlab{c}})\citenamefont {Bennett}, \citenamefont
  {Bernstein}, \citenamefont {Popescu},\ and\ \citenamefont
  {Schumacher}}]{BennettPhysRevA.53.2046}%
  \BibitemOpen
  \bibfield  {author} {\bibinfo {author} {\bibfnamefont {C.~H.}\ \bibnamefont
  {Bennett}}, \bibinfo {author} {\bibfnamefont {H.~J.}\ \bibnamefont
  {Bernstein}}, \bibinfo {author} {\bibfnamefont {S.}~\bibnamefont {Popescu}},\
  and\ \bibinfo {author} {\bibfnamefont {B.}~\bibnamefont {Schumacher}},\
  }\bibfield  {title} {\bibinfo {title} {Concentrating partial entanglement by
  local operations},\ }\href {https://doi.org/10.1103/PhysRevA.53.2046}
  {\bibfield  {journal} {\bibinfo  {journal} {Phys. Rev. A}\ }\textbf {\bibinfo
  {volume} {53}},\ \bibinfo {pages} {2046} (\bibinfo {year}
  {1996}{\natexlab{c}})}\BibitemShut {NoStop}%
\bibitem [{\citenamefont {Horodecki}\ \emph {et~al.}(1998)\citenamefont
  {Horodecki}, \citenamefont {Horodecki},\ and\ \citenamefont
  {Horodecki}}]{PhysRevLett.80.5239}%
  \BibitemOpen
  \bibfield  {author} {\bibinfo {author} {\bibfnamefont {M.}~\bibnamefont
  {Horodecki}}, \bibinfo {author} {\bibfnamefont {P.}~\bibnamefont
  {Horodecki}},\ and\ \bibinfo {author} {\bibfnamefont {R.}~\bibnamefont
  {Horodecki}},\ }\bibfield  {title} {\bibinfo {title} {{Mixed-State
  Entanglement and Distillation: Is there a ``Bound'' Entanglement in
  Nature?}},\ }\href {https://doi.org/10.1103/PhysRevLett.80.5239} {\bibfield
  {journal} {\bibinfo  {journal} {Phys. Rev. Lett.}\ }\textbf {\bibinfo
  {volume} {80}},\ \bibinfo {pages} {5239} (\bibinfo {year}
  {1998})}\BibitemShut {NoStop}%
\bibitem [{\citenamefont {Jonathan}\ and\ \citenamefont
  {Plenio}(1999)}]{JonathanPhysRevLett.83.3566}%
  \BibitemOpen
  \bibfield  {author} {\bibinfo {author} {\bibfnamefont {D.}~\bibnamefont
  {Jonathan}}\ and\ \bibinfo {author} {\bibfnamefont {M.~B.}\ \bibnamefont
  {Plenio}},\ }\bibfield  {title} {\bibinfo {title} {{Entanglement-Assisted
  Local Manipulation of Pure Quantum States}},\ }\href
  {https://doi.org/10.1103/PhysRevLett.83.3566} {\bibfield  {journal} {\bibinfo
   {journal} {Phys. Rev. Lett.}\ }\textbf {\bibinfo {volume} {83}},\ \bibinfo
  {pages} {3566} (\bibinfo {year} {1999})}\BibitemShut {NoStop}%
\bibitem [{\citenamefont {Kondra}\ \emph {et~al.}(2021)\citenamefont {Kondra},
  \citenamefont {Datta},\ and\ \citenamefont
  {Streltsov}}]{PhysRevLett.127.150503}%
  \BibitemOpen
  \bibfield  {author} {\bibinfo {author} {\bibfnamefont {T.~V.}\ \bibnamefont
  {Kondra}}, \bibinfo {author} {\bibfnamefont {C.}~\bibnamefont {Datta}},\ and\
  \bibinfo {author} {\bibfnamefont {A.}~\bibnamefont {Streltsov}},\ }\bibfield
  {title} {\bibinfo {title} {{Catalytic Transformations of Pure Entangled
  States}},\ }\href {https://doi.org/10.1103/PhysRevLett.127.150503} {\bibfield
   {journal} {\bibinfo  {journal} {Phys. Rev. Lett.}\ }\textbf {\bibinfo
  {volume} {127}},\ \bibinfo {pages} {150503} (\bibinfo {year}
  {2021})}\BibitemShut {NoStop}%
\bibitem [{\citenamefont {Lipka-Bartosik}\ and\ \citenamefont
  {Skrzypczyk}(2021)}]{Lipka-Bartosik2102.11846}%
  \BibitemOpen
  \bibfield  {author} {\bibinfo {author} {\bibfnamefont {P.}~\bibnamefont
  {Lipka-Bartosik}}\ and\ \bibinfo {author} {\bibfnamefont {P.}~\bibnamefont
  {Skrzypczyk}},\ }\bibfield  {title} {\bibinfo {title} {{Catalytic Quantum
  Teleportation}},\ }\href {https://doi.org/10.1103/PhysRevLett.127.080502}
  {\bibfield  {journal} {\bibinfo  {journal} {Phys. Rev. Lett.}\ }\textbf
  {\bibinfo {volume} {127}},\ \bibinfo {pages} {080502} (\bibinfo {year}
  {2021})}\BibitemShut {NoStop}%
\bibitem [{\citenamefont {Rubboli}\ and\ \citenamefont
  {Tomamichel}(2022)}]{Rubboli2111.13356}%
  \BibitemOpen
  \bibfield  {author} {\bibinfo {author} {\bibfnamefont {R.}~\bibnamefont
  {Rubboli}}\ and\ \bibinfo {author} {\bibfnamefont {M.}~\bibnamefont
  {Tomamichel}},\ }\bibfield  {title} {\bibinfo {title} {{Fundamental Limits on
  Correlated Catalytic State Transformations}},\ }\href
  {https://doi.org/10.1103/PhysRevLett.129.120506} {\bibfield  {journal}
  {\bibinfo  {journal} {Phys. Rev. Lett.}\ }\textbf {\bibinfo {volume} {129}},\
  \bibinfo {pages} {120506} (\bibinfo {year} {2022})}\BibitemShut {NoStop}%
\bibitem [{\citenamefont {Datta}\ \emph
  {et~al.}(2022{\natexlab{a}})\citenamefont {Datta}, \citenamefont {Kondra},
  \citenamefont {Miller},\ and\ \citenamefont
  {Streltsov}}]{datta2022entanglement}%
  \BibitemOpen
  \bibfield  {author} {\bibinfo {author} {\bibfnamefont {C.}~\bibnamefont
  {Datta}}, \bibinfo {author} {\bibfnamefont {T.~V.}\ \bibnamefont {Kondra}},
  \bibinfo {author} {\bibfnamefont {M.}~\bibnamefont {Miller}},\ and\ \bibinfo
  {author} {\bibfnamefont {A.}~\bibnamefont {Streltsov}},\ }\bibfield  {title}
  {\bibinfo {title} {Entanglement catalysis for quantum states and noisy
  channels},\ }\href {https://arxiv.org/abs/2202.05228} {\bibfield  {journal}
  {\bibinfo  {journal} {arXiv:2202.05228}\ } (\bibinfo {year}
  {2022}{\natexlab{a}})}\BibitemShut {NoStop}%
\bibitem [{\citenamefont {Datta}\ \emph
  {et~al.}(2022{\natexlab{b}})\citenamefont {Datta}, \citenamefont {Kondra},
  \citenamefont {Miller},\ and\ \citenamefont {Streltsov}}]{catalysis_review}%
  \BibitemOpen
  \bibfield  {author} {\bibinfo {author} {\bibfnamefont {C.}~\bibnamefont
  {Datta}}, \bibinfo {author} {\bibfnamefont {T.~V.}\ \bibnamefont {Kondra}},
  \bibinfo {author} {\bibfnamefont {M.}~\bibnamefont {Miller}},\ and\ \bibinfo
  {author} {\bibfnamefont {A.}~\bibnamefont {Streltsov}},\ }\bibfield  {title}
  {\bibinfo {title} {Catalysis of entanglement and other quantum resources},\
  }\href {https://arxiv.org/abs/2207.05694} {\bibfield  {journal} {\bibinfo
  {journal} {arXiv:2207.05694}\ } (\bibinfo {year}
  {2022}{\natexlab{b}})}\BibitemShut {NoStop}%
\bibitem [{\citenamefont {Duan}\ \emph
  {et~al.}(2005{\natexlab{a}})\citenamefont {Duan}, \citenamefont {Feng},\ and\
  \citenamefont {Ying}}]{PhysRevA.72.024306}%
  \BibitemOpen
  \bibfield  {author} {\bibinfo {author} {\bibfnamefont {R.}~\bibnamefont
  {Duan}}, \bibinfo {author} {\bibfnamefont {Y.}~\bibnamefont {Feng}},\ and\
  \bibinfo {author} {\bibfnamefont {M.}~\bibnamefont {Ying}},\ }\bibfield
  {title} {\bibinfo {title} {{Entanglement-assisted transformation is
  asymptotically equivalent to multiple-copy transformation}},\ }\href
  {https://doi.org/10.1103/PhysRevA.72.024306} {\bibfield  {journal} {\bibinfo
  {journal} {Phys. Rev. A}\ }\textbf {\bibinfo {volume} {72}},\ \bibinfo
  {pages} {024306} (\bibinfo {year} {2005}{\natexlab{a}})}\BibitemShut
  {NoStop}%
\bibitem [{\citenamefont {Duan}\ \emph
  {et~al.}(2005{\natexlab{b}})\citenamefont {Duan}, \citenamefont {Feng},
  \citenamefont {Li},\ and\ \citenamefont {Ying}}]{DuanPhysRevA.71.042319}%
  \BibitemOpen
  \bibfield  {author} {\bibinfo {author} {\bibfnamefont {R.}~\bibnamefont
  {Duan}}, \bibinfo {author} {\bibfnamefont {Y.}~\bibnamefont {Feng}}, \bibinfo
  {author} {\bibfnamefont {X.}~\bibnamefont {Li}},\ and\ \bibinfo {author}
  {\bibfnamefont {M.}~\bibnamefont {Ying}},\ }\bibfield  {title} {\bibinfo
  {title} {{Multiple-copy entanglement transformation and entanglement
  catalysis}},\ }\href {https://doi.org/10.1103/PhysRevA.71.042319} {\bibfield
  {journal} {\bibinfo  {journal} {Phys. Rev. A}\ }\textbf {\bibinfo {volume}
  {71}},\ \bibinfo {pages} {042319} (\bibinfo {year}
  {2005}{\natexlab{b}})}\BibitemShut {NoStop}%
\bibitem [{\citenamefont {Shiraishi}\ and\ \citenamefont
  {Sagawa}(2021)}]{shiraishi2020quantum}%
  \BibitemOpen
  \bibfield  {author} {\bibinfo {author} {\bibfnamefont {N.}~\bibnamefont
  {Shiraishi}}\ and\ \bibinfo {author} {\bibfnamefont {T.}~\bibnamefont
  {Sagawa}},\ }\bibfield  {title} {\bibinfo {title} {{Quantum Thermodynamics of
  Correlated-Catalytic State Conversion at Small Scale}},\ }\href
  {https://doi.org/10.1103/PhysRevLett.126.150502} {\bibfield  {journal}
  {\bibinfo  {journal} {Phys. Rev. Lett.}\ }\textbf {\bibinfo {volume} {126}},\
  \bibinfo {pages} {150502} (\bibinfo {year} {2021})}\BibitemShut {NoStop}%
\bibitem [{\citenamefont {Wilming}(2021)}]{wilming2020entropy}%
  \BibitemOpen
  \bibfield  {author} {\bibinfo {author} {\bibfnamefont {H.}~\bibnamefont
  {Wilming}},\ }\bibfield  {title} {\bibinfo {title} {{Entropy and Reversible
  Catalysis}},\ }\href {https://doi.org/10.1103/PhysRevLett.127.260402}
  {\bibfield  {journal} {\bibinfo  {journal} {Phys. Rev. Lett.}\ }\textbf
  {\bibinfo {volume} {127}},\ \bibinfo {pages} {260402} (\bibinfo {year}
  {2021})}\BibitemShut {NoStop}%
\bibitem [{\citenamefont {Horodecki}\ \emph {et~al.}(2002)\citenamefont
  {Horodecki}, \citenamefont {Oppenheim},\ and\ \citenamefont
  {Horodecki}}]{PhysRevLett.89.240403}%
  \BibitemOpen
  \bibfield  {author} {\bibinfo {author} {\bibfnamefont {M.}~\bibnamefont
  {Horodecki}}, \bibinfo {author} {\bibfnamefont {J.}~\bibnamefont
  {Oppenheim}},\ and\ \bibinfo {author} {\bibfnamefont {R.}~\bibnamefont
  {Horodecki}},\ }\bibfield  {title} {\bibinfo {title} {{Are the Laws of
  Entanglement Theory Thermodynamical?}},\ }\href
  {https://doi.org/10.1103/PhysRevLett.89.240403} {\bibfield  {journal}
  {\bibinfo  {journal} {Phys. Rev. Lett.}\ }\textbf {\bibinfo {volume} {89}},\
  \bibinfo {pages} {240403} (\bibinfo {year} {2002})}\BibitemShut {NoStop}%
\bibitem [{\citenamefont {Popescu}\ \emph {et~al.}(2006)\citenamefont
  {Popescu}, \citenamefont {Short},\ and\ \citenamefont
  {Winter}}]{Popescu2006}%
  \BibitemOpen
  \bibfield  {author} {\bibinfo {author} {\bibfnamefont {S.}~\bibnamefont
  {Popescu}}, \bibinfo {author} {\bibfnamefont {A.~J.}\ \bibnamefont {Short}},\
  and\ \bibinfo {author} {\bibfnamefont {A.}~\bibnamefont {Winter}},\
  }\bibfield  {title} {\bibinfo {title} {Entanglement and the foundations of
  statistical mechanics},\ }\href {https://doi.org/10.1038/nphys444} {\bibfield
   {journal} {\bibinfo  {journal} {Nat. Phys.}\ }\textbf {\bibinfo {volume}
  {2}},\ \bibinfo {pages} {754} (\bibinfo {year} {2006})}\BibitemShut {NoStop}%
\bibitem [{\citenamefont {Brand{\~a}o}\ and\ \citenamefont
  {Plenio}(2008)}]{Brandao2008}%
  \BibitemOpen
  \bibfield  {author} {\bibinfo {author} {\bibfnamefont {F.~G. S.~L.}\
  \bibnamefont {Brand{\~a}o}}\ and\ \bibinfo {author} {\bibfnamefont {M.~B.}\
  \bibnamefont {Plenio}},\ }\bibfield  {title} {\bibinfo {title} {Entanglement
  theory and the second law of thermodynamics},\ }\href
  {https://doi.org/10.1038/nphys1100} {\bibfield  {journal} {\bibinfo
  {journal} {Nat. Phys.}\ }\textbf {\bibinfo {volume} {4}},\ \bibinfo {pages}
  {873} (\bibinfo {year} {2008})}\BibitemShut {NoStop}%
\bibitem [{\citenamefont {Lami}\ and\ \citenamefont
  {Regula}(2023{\natexlab{a}})}]{Lami2023}%
  \BibitemOpen
  \bibfield  {author} {\bibinfo {author} {\bibfnamefont {L.}~\bibnamefont
  {Lami}}\ and\ \bibinfo {author} {\bibfnamefont {B.}~\bibnamefont {Regula}},\
  }\bibfield  {title} {\bibinfo {title} {No second law of entanglement
  manipulation after all},\ }\href {https://doi.org/10.1038/s41567-022-01873-9}
  {\bibfield  {journal} {\bibinfo  {journal} {Nat. Phys.}\ }\textbf {\bibinfo
  {volume} {19}},\ \bibinfo {pages} {184} (\bibinfo {year}
  {2023}{\natexlab{a}})}\BibitemShut {NoStop}%
\bibitem [{\citenamefont {Plenio}\ and\ \citenamefont
  {Virmani}(2007)}]{PlenioMeasures}%
  \BibitemOpen
  \bibfield  {author} {\bibinfo {author} {\bibfnamefont {M.~B.}\ \bibnamefont
  {Plenio}}\ and\ \bibinfo {author} {\bibfnamefont {S.}~\bibnamefont
  {Virmani}},\ }\bibfield  {title} {\bibinfo {title} {An introduction to
  entanglement measures},\ }\href {https://arxiv.org/abs/quant-ph/0504163}
  {\bibfield  {journal} {\bibinfo  {journal} {Quant. Inf. Comput.}\ }\textbf
  {\bibinfo {volume} {7}},\ \bibinfo {pages} {1} (\bibinfo {year}
  {2007})}\BibitemShut {NoStop}%
\bibitem [{Note1()}]{Note1}%
  \BibitemOpen
  \bibinfo {note} {See Supplemental Material for technical explanations and
  more details, which contains the additional references~\cite
  {Eggeling_2001,959270,Devetak2005,Matthews_2008,PhysRevA.65.012107,Palazuelos2022,PhysRevA.75.012305,PhysRevLett.102.170503,Christandl_2004,Alicki_2004,RevModPhys.91.025001}}\BibitemShut
  {NoStop}%
\bibitem [{\citenamefont {Bennett}\ \emph {et~al.}(2000)\citenamefont
  {Bennett}, \citenamefont {Popescu}, \citenamefont {Rohrlich}, \citenamefont
  {Smolin},\ and\ \citenamefont {Thapliyal}}]{PhysRevA.63.012307}%
  \BibitemOpen
  \bibfield  {author} {\bibinfo {author} {\bibfnamefont {C.~H.}\ \bibnamefont
  {Bennett}}, \bibinfo {author} {\bibfnamefont {S.}~\bibnamefont {Popescu}},
  \bibinfo {author} {\bibfnamefont {D.}~\bibnamefont {Rohrlich}}, \bibinfo
  {author} {\bibfnamefont {J.~A.}\ \bibnamefont {Smolin}},\ and\ \bibinfo
  {author} {\bibfnamefont {A.~V.}\ \bibnamefont {Thapliyal}},\ }\bibfield
  {title} {\bibinfo {title} {Exact and asymptotic measures of multipartite
  pure-state entanglement},\ }\href
  {https://doi.org/10.1103/PhysRevA.63.012307} {\bibfield  {journal} {\bibinfo
  {journal} {Phys. Rev. A}\ }\textbf {\bibinfo {volume} {63}},\ \bibinfo
  {pages} {012307} (\bibinfo {year} {2000})}\BibitemShut {NoStop}%
\bibitem [{\citenamefont {Eisert}\ and\ \citenamefont
  {Wilkens}(2000)}]{EisertPhysRevLett.85.437}%
  \BibitemOpen
  \bibfield  {author} {\bibinfo {author} {\bibfnamefont {J.}~\bibnamefont
  {Eisert}}\ and\ \bibinfo {author} {\bibfnamefont {M.}~\bibnamefont
  {Wilkens}},\ }\bibfield  {title} {\bibinfo {title} {{Catalysis of
  Entanglement Manipulation for Mixed States}},\ }\href
  {https://doi.org/10.1103/PhysRevLett.85.437} {\bibfield  {journal} {\bibinfo
  {journal} {Phys. Rev. Lett.}\ }\textbf {\bibinfo {volume} {85}},\ \bibinfo
  {pages} {437} (\bibinfo {year} {2000})}\BibitemShut {NoStop}%
\bibitem [{\citenamefont {Brand{\~a}o}\ \emph {et~al.}(2015)\citenamefont
  {Brand{\~a}o}, \citenamefont {Horodecki}, \citenamefont {Ng}, \citenamefont
  {Oppenheim},\ and\ \citenamefont {Wehner}}]{Brandao3275}%
  \BibitemOpen
  \bibfield  {author} {\bibinfo {author} {\bibfnamefont {F.}~\bibnamefont
  {Brand{\~a}o}}, \bibinfo {author} {\bibfnamefont {M.}~\bibnamefont
  {Horodecki}}, \bibinfo {author} {\bibfnamefont {N.}~\bibnamefont {Ng}},
  \bibinfo {author} {\bibfnamefont {J.}~\bibnamefont {Oppenheim}},\ and\
  \bibinfo {author} {\bibfnamefont {S.}~\bibnamefont {Wehner}},\ }\bibfield
  {title} {\bibinfo {title} {The second laws of quantum thermodynamics},\
  }\href {https://doi.org/10.1073/pnas.1411728112} {\bibfield  {journal}
  {\bibinfo  {journal} {Proc. Natl. Acad. Sci. U.S.A.}\ }\textbf {\bibinfo
  {volume} {112}},\ \bibinfo {pages} {3275} (\bibinfo {year}
  {2015})}\BibitemShut {NoStop}%
\bibitem [{\citenamefont {M\"uller}(2018)}]{PhysRevX.8.041051}%
  \BibitemOpen
  \bibfield  {author} {\bibinfo {author} {\bibfnamefont {M.~P.}\ \bibnamefont
  {M\"uller}},\ }\bibfield  {title} {\bibinfo {title} {{Correlating Thermal
  Machines and the Second Law at the Nanoscale}},\ }\href
  {https://doi.org/10.1103/PhysRevX.8.041051} {\bibfield  {journal} {\bibinfo
  {journal} {Phys. Rev. X}\ }\textbf {\bibinfo {volume} {8}},\ \bibinfo {pages}
  {041051} (\bibinfo {year} {2018})}\BibitemShut {NoStop}%
\bibitem [{\citenamefont {Boes}\ \emph {et~al.}(2019)\citenamefont {Boes},
  \citenamefont {Eisert}, \citenamefont {Gallego}, \citenamefont {M\"uller},\
  and\ \citenamefont {Wilming}}]{BoesPhysRevLett.122.210402}%
  \BibitemOpen
  \bibfield  {author} {\bibinfo {author} {\bibfnamefont {P.}~\bibnamefont
  {Boes}}, \bibinfo {author} {\bibfnamefont {J.}~\bibnamefont {Eisert}},
  \bibinfo {author} {\bibfnamefont {R.}~\bibnamefont {Gallego}}, \bibinfo
  {author} {\bibfnamefont {M.~P.}\ \bibnamefont {M\"uller}},\ and\ \bibinfo
  {author} {\bibfnamefont {H.}~\bibnamefont {Wilming}},\ }\bibfield  {title}
  {\bibinfo {title} {{Von Neumann Entropy from Unitarity}},\ }\href
  {https://doi.org/10.1103/PhysRevLett.122.210402} {\bibfield  {journal}
  {\bibinfo  {journal} {Phys. Rev. Lett.}\ }\textbf {\bibinfo {volume} {122}},\
  \bibinfo {pages} {210402} (\bibinfo {year} {2019})}\BibitemShut {NoStop}%
\bibitem [{\citenamefont {Wilming}\ \emph {et~al.}(2017)\citenamefont
  {Wilming}, \citenamefont {Gallego},\ and\ \citenamefont
  {Eisert}}]{Wilminge19060241}%
  \BibitemOpen
  \bibfield  {author} {\bibinfo {author} {\bibfnamefont {H.}~\bibnamefont
  {Wilming}}, \bibinfo {author} {\bibfnamefont {R.}~\bibnamefont {Gallego}},\
  and\ \bibinfo {author} {\bibfnamefont {J.}~\bibnamefont {Eisert}},\
  }\bibfield  {title} {\bibinfo {title} {{Axiomatic Characterization of the
  Quantum Relative Entropy and Free Energy}},\ }\href
  {https://doi.org/10.3390/e19060241} {\bibfield  {journal} {\bibinfo
  {journal} {Entropy}\ }\textbf {\bibinfo {volume} {19}},\ \bibinfo {pages}
  {241} (\bibinfo {year} {2017})}\BibitemShut {NoStop}%
\bibitem [{\citenamefont {Ferrari}\ \emph {et~al.}(2023)\citenamefont
  {Ferrari}, \citenamefont {Lami}, \citenamefont {Theurer},\ and\ \citenamefont
  {Plenio}}]{Ferrari2023}%
  \BibitemOpen
  \bibfield  {author} {\bibinfo {author} {\bibfnamefont {G.}~\bibnamefont
  {Ferrari}}, \bibinfo {author} {\bibfnamefont {L.}~\bibnamefont {Lami}},
  \bibinfo {author} {\bibfnamefont {T.}~\bibnamefont {Theurer}},\ and\ \bibinfo
  {author} {\bibfnamefont {M.~B.}\ \bibnamefont {Plenio}},\ }\bibfield  {title}
  {\bibinfo {title} {{Asymptotic State Transformations of Continuous Variable
  Resources}},\ }\href {https://doi.org/10.1007/s00220-022-04523-6} {\bibfield
  {journal} {\bibinfo  {journal} {Commun. Math. Phys.}\ }\textbf {\bibinfo
  {volume} {398}},\ \bibinfo {pages} {291} (\bibinfo {year}
  {2023})}\BibitemShut {NoStop}%
\bibitem [{\citenamefont {Horodecki}\ \emph {et~al.}(1997)\citenamefont
  {Horodecki}, \citenamefont {Horodecki},\ and\ \citenamefont
  {Horodecki}}]{HorodeckiPhysRevLett.78.574}%
  \BibitemOpen
  \bibfield  {author} {\bibinfo {author} {\bibfnamefont {M.}~\bibnamefont
  {Horodecki}}, \bibinfo {author} {\bibfnamefont {P.}~\bibnamefont
  {Horodecki}},\ and\ \bibinfo {author} {\bibfnamefont {R.}~\bibnamefont
  {Horodecki}},\ }\bibfield  {title} {\bibinfo {title} {{Inseparable Two Spin-
  $\frac{1}{2}$ Density Matrices Can Be Distilled to a Singlet Form}},\ }\href
  {https://doi.org/10.1103/PhysRevLett.78.574} {\bibfield  {journal} {\bibinfo
  {journal} {Phys. Rev. Lett.}\ }\textbf {\bibinfo {volume} {78}},\ \bibinfo
  {pages} {574} (\bibinfo {year} {1997})}\BibitemShut {NoStop}%
\bibitem [{\citenamefont {Horodecki}\ \emph {et~al.}(1999)\citenamefont
  {Horodecki}, \citenamefont {Horodecki},\ and\ \citenamefont
  {Horodecki}}]{PhysRevLett.82.1056}%
  \BibitemOpen
  \bibfield  {author} {\bibinfo {author} {\bibfnamefont {P.}~\bibnamefont
  {Horodecki}}, \bibinfo {author} {\bibfnamefont {M.}~\bibnamefont
  {Horodecki}},\ and\ \bibinfo {author} {\bibfnamefont {R.}~\bibnamefont
  {Horodecki}},\ }\bibfield  {title} {\bibinfo {title} {{Bound Entanglement Can
  Be Activated}},\ }\href {https://doi.org/10.1103/PhysRevLett.82.1056}
  {\bibfield  {journal} {\bibinfo  {journal} {Phys. Rev. Lett.}\ }\textbf
  {\bibinfo {volume} {82}},\ \bibinfo {pages} {1056} (\bibinfo {year}
  {1999})}\BibitemShut {NoStop}%
\bibitem [{\citenamefont {Shor}\ \emph {et~al.}(2001)\citenamefont {Shor},
  \citenamefont {Smolin},\ and\ \citenamefont {Terhal}}]{PhysRevLett.86.2681}%
  \BibitemOpen
  \bibfield  {author} {\bibinfo {author} {\bibfnamefont {P.~W.}\ \bibnamefont
  {Shor}}, \bibinfo {author} {\bibfnamefont {J.~A.}\ \bibnamefont {Smolin}},\
  and\ \bibinfo {author} {\bibfnamefont {B.~M.}\ \bibnamefont {Terhal}},\
  }\bibfield  {title} {\bibinfo {title} {{Nonadditivity of Bipartite
  Distillable Entanglement Follows from a Conjecture on Bound Entangled Werner
  States}},\ }\href {https://doi.org/10.1103/PhysRevLett.86.2681} {\bibfield
  {journal} {\bibinfo  {journal} {Phys. Rev. Lett.}\ }\textbf {\bibinfo
  {volume} {86}},\ \bibinfo {pages} {2681} (\bibinfo {year}
  {2001})}\BibitemShut {NoStop}%
\bibitem [{\citenamefont {V{\'e}rtesi}\ and\ \citenamefont
  {Brunner}(2014)}]{Vertesi2014}%
  \BibitemOpen
  \bibfield  {author} {\bibinfo {author} {\bibfnamefont {T.}~\bibnamefont
  {V{\'e}rtesi}}\ and\ \bibinfo {author} {\bibfnamefont {N.}~\bibnamefont
  {Brunner}},\ }\bibfield  {title} {\bibinfo {title} {{Disproving the Peres
  conjecture by showing Bell nonlocality from bound entanglement}},\ }\href
  {https://doi.org/10.1038/ncomms6297} {\bibfield  {journal} {\bibinfo
  {journal} {Nat. Commun.}\ }\textbf {\bibinfo {volume} {5}},\ \bibinfo {pages}
  {5297} (\bibinfo {year} {2014})}\BibitemShut {NoStop}%
\bibitem [{\citenamefont {Smolin}\ \emph {et~al.}(2005)\citenamefont {Smolin},
  \citenamefont {Verstraete},\ and\ \citenamefont
  {Winter}}]{SmolinPhysRevA.72.052317}%
  \BibitemOpen
  \bibfield  {author} {\bibinfo {author} {\bibfnamefont {J.~A.}\ \bibnamefont
  {Smolin}}, \bibinfo {author} {\bibfnamefont {F.}~\bibnamefont {Verstraete}},\
  and\ \bibinfo {author} {\bibfnamefont {A.}~\bibnamefont {Winter}},\
  }\bibfield  {title} {\bibinfo {title} {Entanglement of assistance and
  multipartite state distillation},\ }\href
  {https://doi.org/10.1103/PhysRevA.72.052317} {\bibfield  {journal} {\bibinfo
  {journal} {Phys. Rev. A}\ }\textbf {\bibinfo {volume} {72}},\ \bibinfo
  {pages} {052317} (\bibinfo {year} {2005})}\BibitemShut {NoStop}%
\bibitem [{\citenamefont {Horodecki}\ \emph {et~al.}(2005)\citenamefont
  {Horodecki}, \citenamefont {Oppenheim},\ and\ \citenamefont {Winter}}]{qsm}%
  \BibitemOpen
  \bibfield  {author} {\bibinfo {author} {\bibfnamefont {M.}~\bibnamefont
  {Horodecki}}, \bibinfo {author} {\bibfnamefont {J.}~\bibnamefont
  {Oppenheim}},\ and\ \bibinfo {author} {\bibfnamefont {A.}~\bibnamefont
  {Winter}},\ }\bibfield  {title} {\bibinfo {title} {Partial quantum
  information},\ }\href {https://doi.org/10.1038/nature03909} {\bibfield
  {journal} {\bibinfo  {journal} {Nature}\ }\textbf {\bibinfo {volume} {436}},\
  \bibinfo {pages} {673} (\bibinfo {year} {2005})}\BibitemShut {NoStop}%
\bibitem [{\citenamefont {G\"uhne}\ and\ \citenamefont
  {T\'oth}(2009)}]{GUHNE20091}%
  \BibitemOpen
  \bibfield  {author} {\bibinfo {author} {\bibfnamefont {O.}~\bibnamefont
  {G\"uhne}}\ and\ \bibinfo {author} {\bibfnamefont {G.}~\bibnamefont
  {T\'oth}},\ }\bibfield  {title} {\bibinfo {title} {Entanglement detection},\
  }\href {https://doi.org/https://doi.org/10.1016/j.physrep.2009.02.004}
  {\bibfield  {journal} {\bibinfo  {journal} {Phys. Rep.}\ }\textbf {\bibinfo
  {volume} {474}},\ \bibinfo {pages} {1} (\bibinfo {year} {2009})}\BibitemShut
  {NoStop}%
\bibitem [{\citenamefont {Rains}(1999{\natexlab{a}})}]{Rains_1999}%
  \BibitemOpen
  \bibfield  {author} {\bibinfo {author} {\bibfnamefont {E.~M.}\ \bibnamefont
  {Rains}},\ }\bibfield  {title} {\bibinfo {title} {Bound on distillable
  entanglement},\ }\href {https://doi.org/10.1103/PhysRevA.60.179} {\bibfield
  {journal} {\bibinfo  {journal} {Phys. Rev. A}\ }\textbf {\bibinfo {volume}
  {60}},\ \bibinfo {pages} {179} (\bibinfo {year}
  {1999}{\natexlab{a}})}\BibitemShut {NoStop}%
\bibitem [{\citenamefont {Rains}(1999{\natexlab{b}})}]{Rains_1999a}%
  \BibitemOpen
  \bibfield  {author} {\bibinfo {author} {\bibfnamefont {E.~M.}\ \bibnamefont
  {Rains}},\ }\bibfield  {title} {\bibinfo {title} {Rigorous treatment of
  distillable entanglement},\ }\href {https://doi.org/10.1103/PhysRevA.60.173}
  {\bibfield  {journal} {\bibinfo  {journal} {Phys. Rev. A}\ }\textbf {\bibinfo
  {volume} {60}},\ \bibinfo {pages} {173} (\bibinfo {year}
  {1999}{\natexlab{b}})}\BibitemShut {NoStop}%
\bibitem [{\citenamefont {Lami}\ \emph {et~al.}(2024)\citenamefont {Lami},
  \citenamefont {Regula},\ and\ \citenamefont {Streltsov}}]{Lami2023b}%
  \BibitemOpen
  \bibfield  {author} {\bibinfo {author} {\bibfnamefont {L.}~\bibnamefont
  {Lami}}, \bibinfo {author} {\bibfnamefont {B.}~\bibnamefont {Regula}},\ and\
  \bibinfo {author} {\bibfnamefont {A.}~\bibnamefont {Streltsov}},\ }\bibfield
  {title} {\bibinfo {title} {No-go theorem for entanglement distillation using
  catalysis},\ }\href {https://doi.org/10.1103/PhysRevA.109.L050401} {\bibfield
   {journal} {\bibinfo  {journal} {Phys. Rev. A}\ }\textbf {\bibinfo {volume}
  {109}},\ \bibinfo {pages} {L050401} (\bibinfo {year} {2024})}\BibitemShut
  {NoStop}%
\bibitem [{\citenamefont {DiVincenzo}\ \emph {et~al.}(2000)\citenamefont
  {DiVincenzo}, \citenamefont {Shor}, \citenamefont {Smolin}, \citenamefont
  {Terhal},\ and\ \citenamefont {Thapliyal}}]{PhysRevA.61.062312}%
  \BibitemOpen
  \bibfield  {author} {\bibinfo {author} {\bibfnamefont {D.~P.}\ \bibnamefont
  {DiVincenzo}}, \bibinfo {author} {\bibfnamefont {P.~W.}\ \bibnamefont
  {Shor}}, \bibinfo {author} {\bibfnamefont {J.~A.}\ \bibnamefont {Smolin}},
  \bibinfo {author} {\bibfnamefont {B.~M.}\ \bibnamefont {Terhal}},\ and\
  \bibinfo {author} {\bibfnamefont {A.~V.}\ \bibnamefont {Thapliyal}},\
  }\bibfield  {title} {\bibinfo {title} {Evidence for bound entangled states
  with negative partial transpose},\ }\href
  {https://doi.org/10.1103/PhysRevA.61.062312} {\bibfield  {journal} {\bibinfo
  {journal} {Phys. Rev. A}\ }\textbf {\bibinfo {volume} {61}},\ \bibinfo
  {pages} {062312} (\bibinfo {year} {2000})}\BibitemShut {NoStop}%
\bibitem [{\citenamefont {Pankowski}\ \emph {et~al.}(2010)\citenamefont
  {Pankowski}, \citenamefont {Piani}, \citenamefont {Horodecki},\ and\
  \citenamefont {Horodecki}}]{5508622}%
  \BibitemOpen
  \bibfield  {author} {\bibinfo {author} {\bibfnamefont {{\L}.}~\bibnamefont
  {Pankowski}}, \bibinfo {author} {\bibfnamefont {M.}~\bibnamefont {Piani}},
  \bibinfo {author} {\bibfnamefont {M.}~\bibnamefont {Horodecki}},\ and\
  \bibinfo {author} {\bibfnamefont {P.}~\bibnamefont {Horodecki}},\ }\bibfield
  {title} {\bibinfo {title} {{A Few Steps More Towards NPT Bound
  Entanglement}},\ }\href {https://doi.org/10.1109/TIT.2010.2050810} {\bibfield
   {journal} {\bibinfo  {journal} {IEEE Trans. Inf. Theory}\ }\textbf {\bibinfo
  {volume} {56}},\ \bibinfo {pages} {4085} (\bibinfo {year}
  {2010})}\BibitemShut {NoStop}%
\bibitem [{\citenamefont {Chitambar}\ \emph {et~al.}(2020)\citenamefont
  {Chitambar}, \citenamefont {de~Vicente}, \citenamefont {Girard},\ and\
  \citenamefont {Gour}}]{Chitambar_2020}%
  \BibitemOpen
  \bibfield  {author} {\bibinfo {author} {\bibfnamefont {E.}~\bibnamefont
  {Chitambar}}, \bibinfo {author} {\bibfnamefont {J.~I.}\ \bibnamefont
  {de~Vicente}}, \bibinfo {author} {\bibfnamefont {M.~W.}\ \bibnamefont
  {Girard}},\ and\ \bibinfo {author} {\bibfnamefont {G.}~\bibnamefont {Gour}},\
  }\bibfield  {title} {\bibinfo {title} {{Entanglement manipulation beyond
  local operations and classical communication}},\ }\href
  {https://doi.org/10.1063/1.5124109} {\bibfield  {journal} {\bibinfo
  {journal} {J. Math. Phys.}\ }\textbf {\bibinfo {volume} {61}},\ \bibinfo
  {pages} {042201} (\bibinfo {year} {2020})}\BibitemShut {NoStop}%
\bibitem [{\citenamefont {Lami}\ and\ \citenamefont
  {Regula}(2023{\natexlab{b}})}]{lami2023c}%
  \BibitemOpen
  \bibfield  {author} {\bibinfo {author} {\bibfnamefont {L.}~\bibnamefont
  {Lami}}\ and\ \bibinfo {author} {\bibfnamefont {B.}~\bibnamefont {Regula}},\
  }\bibfield  {title} {\bibinfo {title} {Distillable entanglement under dually
  non-entangling operations},\ }\href {https://arxiv.org/abs/2307.11008}
  {\bibfield  {journal} {\bibinfo  {journal} {arXiv:2307.11008}\ } (\bibinfo
  {year} {2023}{\natexlab{b}})}\BibitemShut {NoStop}%
\bibitem [{\citenamefont {de~Vicente}(2024)}]{PhysRevLett.133.050202}%
  \BibitemOpen
  \bibfield  {author} {\bibinfo {author} {\bibfnamefont {J.~I.}\ \bibnamefont
  {de~Vicente}},\ }\bibfield  {title} {\bibinfo {title} {{Maximally Entangled
  Mixed States for a Fixed Spectrum Do Not Always Exist}},\ }\href
  {https://doi.org/10.1103/PhysRevLett.133.050202} {\bibfield  {journal}
  {\bibinfo  {journal} {Phys. Rev. Lett.}\ }\textbf {\bibinfo {volume} {133}},\
  \bibinfo {pages} {050202} (\bibinfo {year} {2024})}\BibitemShut {NoStop}%
\bibitem [{\citenamefont {Streltsov}\ \emph {et~al.}(2020)\citenamefont
  {Streltsov}, \citenamefont {Meignant},\ and\ \citenamefont
  {Eisert}}]{StreltsovPhysRevLett.125.080502}%
  \BibitemOpen
  \bibfield  {author} {\bibinfo {author} {\bibfnamefont {A.}~\bibnamefont
  {Streltsov}}, \bibinfo {author} {\bibfnamefont {C.}~\bibnamefont
  {Meignant}},\ and\ \bibinfo {author} {\bibfnamefont {J.}~\bibnamefont
  {Eisert}},\ }\bibfield  {title} {\bibinfo {title} {{Rates of Multipartite
  Entanglement Transformations}},\ }\href
  {https://doi.org/10.1103/PhysRevLett.125.080502} {\bibfield  {journal}
  {\bibinfo  {journal} {Phys. Rev. Lett.}\ }\textbf {\bibinfo {volume} {125}},\
  \bibinfo {pages} {080502} (\bibinfo {year} {2020})}\BibitemShut {NoStop}%
\bibitem [{\citenamefont {Cleve}\ \emph {et~al.}(1999)\citenamefont {Cleve},
  \citenamefont {Gottesman},\ and\ \citenamefont {Lo}}]{PhysRevLett.83.648}%
  \BibitemOpen
  \bibfield  {author} {\bibinfo {author} {\bibfnamefont {R.}~\bibnamefont
  {Cleve}}, \bibinfo {author} {\bibfnamefont {D.}~\bibnamefont {Gottesman}},\
  and\ \bibinfo {author} {\bibfnamefont {H.-K.}\ \bibnamefont {Lo}},\
  }\bibfield  {title} {\bibinfo {title} {{How to Share a Quantum Secret}},\
  }\href {https://doi.org/10.1103/PhysRevLett.83.648} {\bibfield  {journal}
  {\bibinfo  {journal} {Phys. Rev. Lett.}\ }\textbf {\bibinfo {volume} {83}},\
  \bibinfo {pages} {648} (\bibinfo {year} {1999})}\BibitemShut {NoStop}%
\bibitem [{\citenamefont {Hillery}\ \emph {et~al.}(1999)\citenamefont
  {Hillery}, \citenamefont {Bu\ifmmode~\check{z}\else \v{z}\fi{}ek},\ and\
  \citenamefont {Berthiaume}}]{PhysRevA.59.1829}%
  \BibitemOpen
  \bibfield  {author} {\bibinfo {author} {\bibfnamefont {M.}~\bibnamefont
  {Hillery}}, \bibinfo {author} {\bibfnamefont {V.}~\bibnamefont
  {Bu\ifmmode~\check{z}\else \v{z}\fi{}ek}},\ and\ \bibinfo {author}
  {\bibfnamefont {A.}~\bibnamefont {Berthiaume}},\ }\bibfield  {title}
  {\bibinfo {title} {Quantum secret sharing},\ }\href
  {https://doi.org/10.1103/PhysRevA.59.1829} {\bibfield  {journal} {\bibinfo
  {journal} {Phys. Rev. A}\ }\textbf {\bibinfo {volume} {59}},\ \bibinfo
  {pages} {1829} (\bibinfo {year} {1999})}\BibitemShut {NoStop}%
\bibitem [{\citenamefont {Eggeling}\ \emph {et~al.}(2001)\citenamefont
  {Eggeling}, \citenamefont {Vollbrecht}, \citenamefont {Werner},\ and\
  \citenamefont {Wolf}}]{Eggeling_2001}%
  \BibitemOpen
  \bibfield  {author} {\bibinfo {author} {\bibfnamefont {T.}~\bibnamefont
  {Eggeling}}, \bibinfo {author} {\bibfnamefont {K.~G.~H.}\ \bibnamefont
  {Vollbrecht}}, \bibinfo {author} {\bibfnamefont {R.~F.}\ \bibnamefont
  {Werner}},\ and\ \bibinfo {author} {\bibfnamefont {M.~M.}\ \bibnamefont
  {Wolf}},\ }\bibfield  {title} {\bibinfo {title} {{Distillability via
  Protocols Respecting the Positivity of Partial Transpose}},\ }\href
  {https://doi.org/10.1103/PhysRevLett.87.257902} {\bibfield  {journal}
  {\bibinfo  {journal} {Phys. Rev. Lett.}\ }\textbf {\bibinfo {volume} {87}},\
  \bibinfo {pages} {257902} (\bibinfo {year} {2001})}\BibitemShut {NoStop}%
\bibitem [{\citenamefont {Rains}(2001)}]{959270}%
  \BibitemOpen
  \bibfield  {author} {\bibinfo {author} {\bibfnamefont {E.}~\bibnamefont
  {Rains}},\ }\bibfield  {title} {\bibinfo {title} {A semidefinite program for
  distillable entanglement},\ }\href {https://doi.org/10.1109/18.959270}
  {\bibfield  {journal} {\bibinfo  {journal} {IEEE Trans. Inf. Theory}\
  }\textbf {\bibinfo {volume} {47}},\ \bibinfo {pages} {2921} (\bibinfo {year}
  {2001})}\BibitemShut {NoStop}%
\bibitem [{\citenamefont {Devetak}\ and\ \citenamefont
  {Winter}(2005)}]{Devetak2005}%
  \BibitemOpen
  \bibfield  {author} {\bibinfo {author} {\bibfnamefont {I.}~\bibnamefont
  {Devetak}}\ and\ \bibinfo {author} {\bibfnamefont {A.}~\bibnamefont
  {Winter}},\ }\bibfield  {title} {\bibinfo {title} {Distillation of secret key
  and entanglement from quantum states},\ }\href
  {https://doi.org/10.1098/rspa.2004.1372} {\bibfield  {journal} {\bibinfo
  {journal} {Proc. R. Soc. Lond. A}\ }\textbf {\bibinfo {volume} {461}},\
  \bibinfo {pages} {207} (\bibinfo {year} {2005})}\BibitemShut {NoStop}%
\bibitem [{\citenamefont {Matthews}\ and\ \citenamefont
  {Winter}(2008)}]{Matthews_2008}%
  \BibitemOpen
  \bibfield  {author} {\bibinfo {author} {\bibfnamefont {W.}~\bibnamefont
  {Matthews}}\ and\ \bibinfo {author} {\bibfnamefont {A.}~\bibnamefont
  {Winter}},\ }\bibfield  {title} {\bibinfo {title} {Pure-state transformations
  and catalysis under operations that completely preserve positivity of partial
  transpose},\ }\href {https://doi.org/10.1103/PhysRevA.78.012317} {\bibfield
  {journal} {\bibinfo  {journal} {Phys. Rev. A}\ }\textbf {\bibinfo {volume}
  {78}},\ \bibinfo {pages} {012317} (\bibinfo {year} {2008})}\BibitemShut
  {NoStop}%
\bibitem [{\citenamefont {Seevinck}\ and\ \citenamefont
  {Uffink}(2001)}]{PhysRevA.65.012107}%
  \BibitemOpen
  \bibfield  {author} {\bibinfo {author} {\bibfnamefont {M.}~\bibnamefont
  {Seevinck}}\ and\ \bibinfo {author} {\bibfnamefont {J.}~\bibnamefont
  {Uffink}},\ }\bibfield  {title} {\bibinfo {title} {Sufficient conditions for
  three-particle entanglement and their tests in recent experiments},\ }\href
  {https://doi.org/10.1103/PhysRevA.65.012107} {\bibfield  {journal} {\bibinfo
  {journal} {Phys. Rev. A}\ }\textbf {\bibinfo {volume} {65}},\ \bibinfo
  {pages} {012107} (\bibinfo {year} {2001})}\BibitemShut {NoStop}%
\bibitem [{\citenamefont {Palazuelos}\ and\ \citenamefont
  {Vicente}(2022)}]{Palazuelos2022}%
  \BibitemOpen
  \bibfield  {author} {\bibinfo {author} {\bibfnamefont {C.}~\bibnamefont
  {Palazuelos}}\ and\ \bibinfo {author} {\bibfnamefont {J.~I.~d.}\ \bibnamefont
  {Vicente}},\ }\bibfield  {title} {\bibinfo {title} {Genuine multipartite
  entanglement of quantum states in the multiple-copy scenario},\ }\href
  {https://doi.org/10.22331/q-2022-06-13-735} {\bibfield  {journal} {\bibinfo
  {journal} {{Quantum}}\ }\textbf {\bibinfo {volume} {6}},\ \bibinfo {pages}
  {735} (\bibinfo {year} {2022})}\BibitemShut {NoStop}%
\bibitem [{\citenamefont {Piani}\ and\ \citenamefont
  {Mora}(2007)}]{PhysRevA.75.012305}%
  \BibitemOpen
  \bibfield  {author} {\bibinfo {author} {\bibfnamefont {M.}~\bibnamefont
  {Piani}}\ and\ \bibinfo {author} {\bibfnamefont {C.~E.}\ \bibnamefont
  {Mora}},\ }\bibfield  {title} {\bibinfo {title} {Class of
  positive-partial-transpose bound entangled states associated with almost any
  set of pure entangled states},\ }\href
  {https://doi.org/10.1103/PhysRevA.75.012305} {\bibfield  {journal} {\bibinfo
  {journal} {Phys. Rev. A}\ }\textbf {\bibinfo {volume} {75}},\ \bibinfo
  {pages} {012305} (\bibinfo {year} {2007})}\BibitemShut {NoStop}%
\bibitem [{\citenamefont {T\'oth}\ and\ \citenamefont
  {G\"uhne}(2009)}]{PhysRevLett.102.170503}%
  \BibitemOpen
  \bibfield  {author} {\bibinfo {author} {\bibfnamefont {G.}~\bibnamefont
  {T\'oth}}\ and\ \bibinfo {author} {\bibfnamefont {O.}~\bibnamefont
  {G\"uhne}},\ }\bibfield  {title} {\bibinfo {title} {{Entanglement and
  Permutational Symmetry}},\ }\href
  {https://doi.org/10.1103/PhysRevLett.102.170503} {\bibfield  {journal}
  {\bibinfo  {journal} {Phys. Rev. Lett.}\ }\textbf {\bibinfo {volume} {102}},\
  \bibinfo {pages} {170503} (\bibinfo {year} {2009})}\BibitemShut {NoStop}%
\bibitem [{\citenamefont {Christandl}\ and\ \citenamefont
  {Winter}(2004)}]{Christandl_2004}%
  \BibitemOpen
  \bibfield  {author} {\bibinfo {author} {\bibfnamefont {M.}~\bibnamefont
  {Christandl}}\ and\ \bibinfo {author} {\bibfnamefont {A.}~\bibnamefont
  {Winter}},\ }\bibfield  {title} {\bibinfo {title} {{``Squashed
  entanglement'': An additive entanglement measure}},\ }\href
  {https://doi.org/10.1063/1.1643788} {\bibfield  {journal} {\bibinfo
  {journal} {J. Math. Phys.}\ }\textbf {\bibinfo {volume} {45}},\ \bibinfo
  {pages} {829} (\bibinfo {year} {2004})}\BibitemShut {NoStop}%
\bibitem [{\citenamefont {Alicki}\ and\ \citenamefont
  {Fannes}(2004)}]{Alicki_2004}%
  \BibitemOpen
  \bibfield  {author} {\bibinfo {author} {\bibfnamefont {R.}~\bibnamefont
  {Alicki}}\ and\ \bibinfo {author} {\bibfnamefont {M.}~\bibnamefont
  {Fannes}},\ }\bibfield  {title} {\bibinfo {title} {Continuity of quantum
  conditional information},\ }\href
  {https://doi.org/10.1088/0305-4470/37/5/l01} {\bibfield  {journal} {\bibinfo
  {journal} {J. Phys. A}\ }\textbf {\bibinfo {volume} {37}},\ \bibinfo {pages}
  {L55} (\bibinfo {year} {2004})}\BibitemShut {NoStop}%
\bibitem [{\citenamefont {Chitambar}\ and\ \citenamefont
  {Gour}(2019)}]{RevModPhys.91.025001}%
  \BibitemOpen
  \bibfield  {author} {\bibinfo {author} {\bibfnamefont {E.}~\bibnamefont
  {Chitambar}}\ and\ \bibinfo {author} {\bibfnamefont {G.}~\bibnamefont
  {Gour}},\ }\bibfield  {title} {\bibinfo {title} {Quantum resource theories},\
  }\href {https://doi.org/10.1103/RevModPhys.91.025001} {\bibfield  {journal}
  {\bibinfo  {journal} {Rev. Mod. Phys.}\ }\textbf {\bibinfo {volume} {91}},\
  \bibinfo {pages} {025001} (\bibinfo {year} {2019})}\BibitemShut {NoStop}%
\end{thebibliography}%


\begin{thebibliography}{25}%
\makeatletter
\providecommand \@ifxundefined [1]{%
 \@ifx{#1\undefined}
}%
\providecommand \@ifnum [1]{%
 \ifnum #1\expandafter \@firstoftwo
 \else \expandafter \@secondoftwo
 \fi
}%
\providecommand \@ifx [1]{%
 \ifx #1\expandafter \@firstoftwo
 \else \expandafter \@secondoftwo
 \fi
}%
\providecommand \natexlab [1]{#1}%
\providecommand \enquote  [1]{``#1''}%
\providecommand \bibnamefont  [1]{#1}%
\providecommand \bibfnamefont [1]{#1}%
\providecommand \citenamefont [1]{#1}%
\providecommand \href@noop [0]{\@secondoftwo}%
\providecommand \href [0]{\begingroup \@sanitize@url \@href}%
\providecommand \@href[1]{\@@startlink{#1}\@@href}%
\providecommand \@@href[1]{\endgroup#1\@@endlink}%
\providecommand \@sanitize@url [0]{\catcode `\\12\catcode `\$12\catcode
  `\&12\catcode `\#12\catcode `\^12\catcode `\_12\catcode `\%12\relax}%
\providecommand \@@startlink[1]{}%
\providecommand \@@endlink[0]{}%
\providecommand \url  [0]{\begingroup\@sanitize@url \@url }%
\providecommand \@url [1]{\endgroup\@href {#1}{\urlprefix }}%
\providecommand \urlprefix  [0]{URL }%
\providecommand \Eprint [0]{\href }%
\providecommand \doibase [0]{https://doi.org/}%
\providecommand \selectlanguage [0]{\@gobble}%
\providecommand \bibinfo  [0]{\@secondoftwo}%
\providecommand \bibfield  [0]{\@secondoftwo}%
\providecommand \translation [1]{[#1]}%
\providecommand \BibitemOpen [0]{}%
\providecommand \bibitemStop [0]{}%
\providecommand \bibitemNoStop [0]{.\EOS\space}%
\providecommand \EOS [0]{\spacefactor3000\relax}%
\providecommand \BibitemShut  [1]{\csname bibitem#1\endcsname}%
\let\auto@bib@innerbib\@empty
\bibitem [{\citenamefont {Rains}(1999{\natexlab{a}})}]{Rains_1999}%
  \BibitemOpen
  \bibfield  {author} {\bibinfo {author} {\bibfnamefont {E.~M.}\ \bibnamefont
  {Rains}},\ }\bibfield  {title} {\bibinfo {title} {Bound on distillable
  entanglement},\ }\href {https://doi.org/10.1103/PhysRevA.60.179} {\bibfield
  {journal} {\bibinfo  {journal} {Phys. Rev. A}\ }\textbf {\bibinfo {volume}
  {60}},\ \bibinfo {pages} {179} (\bibinfo {year}
  {1999}{\natexlab{a}})}\BibitemShut {NoStop}%
\bibitem [{\citenamefont {Rains}(1999{\natexlab{b}})}]{Rains_1999a}%
  \BibitemOpen
  \bibfield  {author} {\bibinfo {author} {\bibfnamefont {E.~M.}\ \bibnamefont
  {Rains}},\ }\bibfield  {title} {\bibinfo {title} {Rigorous treatment of
  distillable entanglement},\ }\href {https://doi.org/10.1103/PhysRevA.60.173}
  {\bibfield  {journal} {\bibinfo  {journal} {Phys. Rev. A}\ }\textbf {\bibinfo
  {volume} {60}},\ \bibinfo {pages} {173} (\bibinfo {year}
  {1999}{\natexlab{b}})}\BibitemShut {NoStop}%
\bibitem [{\citenamefont {Chitambar}\ \emph {et~al.}(2020)\citenamefont
  {Chitambar}, \citenamefont {de~Vicente}, \citenamefont {Girard},\ and\
  \citenamefont {Gour}}]{Chitambar_2020}%
  \BibitemOpen
  \bibfield  {author} {\bibinfo {author} {\bibfnamefont {E.}~\bibnamefont
  {Chitambar}}, \bibinfo {author} {\bibfnamefont {J.~I.}\ \bibnamefont
  {de~Vicente}}, \bibinfo {author} {\bibfnamefont {M.~W.}\ \bibnamefont
  {Girard}},\ and\ \bibinfo {author} {\bibfnamefont {G.}~\bibnamefont {Gour}},\
  }\bibfield  {title} {\bibinfo {title} {{Entanglement manipulation beyond
  local operations and classical communication}},\ }\href
  {https://doi.org/10.1063/1.5124109} {\bibfield  {journal} {\bibinfo
  {journal} {J. Math. Phys.}\ }\textbf {\bibinfo {volume} {61}},\ \bibinfo
  {pages} {042201} (\bibinfo {year} {2020})}\BibitemShut {NoStop}%
\bibitem [{\citenamefont {Eggeling}\ \emph {et~al.}(2001)\citenamefont
  {Eggeling}, \citenamefont {Vollbrecht}, \citenamefont {Werner},\ and\
  \citenamefont {Wolf}}]{Eggeling_2001}%
  \BibitemOpen
  \bibfield  {author} {\bibinfo {author} {\bibfnamefont {T.}~\bibnamefont
  {Eggeling}}, \bibinfo {author} {\bibfnamefont {K.~G.~H.}\ \bibnamefont
  {Vollbrecht}}, \bibinfo {author} {\bibfnamefont {R.~F.}\ \bibnamefont
  {Werner}},\ and\ \bibinfo {author} {\bibfnamefont {M.~M.}\ \bibnamefont
  {Wolf}},\ }\bibfield  {title} {\bibinfo {title} {{Distillability via
  Protocols Respecting the Positivity of Partial Transpose}},\ }\href
  {https://doi.org/10.1103/PhysRevLett.87.257902} {\bibfield  {journal}
  {\bibinfo  {journal} {Phys. Rev. Lett.}\ }\textbf {\bibinfo {volume} {87}},\
  \bibinfo {pages} {257902} (\bibinfo {year} {2001})}\BibitemShut {NoStop}%
\bibitem [{\citenamefont {Kondra}\ \emph {et~al.}(2021)\citenamefont {Kondra},
  \citenamefont {Datta},\ and\ \citenamefont
  {Streltsov}}]{PhysRevLett.127.150503}%
  \BibitemOpen
  \bibfield  {author} {\bibinfo {author} {\bibfnamefont {T.~V.}\ \bibnamefont
  {Kondra}}, \bibinfo {author} {\bibfnamefont {C.}~\bibnamefont {Datta}},\ and\
  \bibinfo {author} {\bibfnamefont {A.}~\bibnamefont {Streltsov}},\ }\bibfield
  {title} {\bibinfo {title} {{Catalytic Transformations of Pure Entangled
  States}},\ }\href {https://doi.org/10.1103/PhysRevLett.127.150503} {\bibfield
   {journal} {\bibinfo  {journal} {Phys. Rev. Lett.}\ }\textbf {\bibinfo
  {volume} {127}},\ \bibinfo {pages} {150503} (\bibinfo {year}
  {2021})}\BibitemShut {NoStop}%
\bibitem [{\citenamefont {Lipka-Bartosik}\ and\ \citenamefont
  {Skrzypczyk}(2021)}]{Lipka-Bartosik2102.11846}%
  \BibitemOpen
  \bibfield  {author} {\bibinfo {author} {\bibfnamefont {P.}~\bibnamefont
  {Lipka-Bartosik}}\ and\ \bibinfo {author} {\bibfnamefont {P.}~\bibnamefont
  {Skrzypczyk}},\ }\bibfield  {title} {\bibinfo {title} {{Catalytic Quantum
  Teleportation}},\ }\href {https://doi.org/10.1103/PhysRevLett.127.080502}
  {\bibfield  {journal} {\bibinfo  {journal} {Phys. Rev. Lett.}\ }\textbf
  {\bibinfo {volume} {127}},\ \bibinfo {pages} {080502} (\bibinfo {year}
  {2021})}\BibitemShut {NoStop}%
\bibitem [{\citenamefont {Duan}\ \emph {et~al.}(2005)\citenamefont {Duan},
  \citenamefont {Feng}, \citenamefont {Li},\ and\ \citenamefont
  {Ying}}]{DuanPhysRevA.71.042319}%
  \BibitemOpen
  \bibfield  {author} {\bibinfo {author} {\bibfnamefont {R.}~\bibnamefont
  {Duan}}, \bibinfo {author} {\bibfnamefont {Y.}~\bibnamefont {Feng}}, \bibinfo
  {author} {\bibfnamefont {X.}~\bibnamefont {Li}},\ and\ \bibinfo {author}
  {\bibfnamefont {M.}~\bibnamefont {Ying}},\ }\bibfield  {title} {\bibinfo
  {title} {{Multiple-copy entanglement transformation and entanglement
  catalysis}},\ }\href {https://doi.org/10.1103/PhysRevA.71.042319} {\bibfield
  {journal} {\bibinfo  {journal} {Phys. Rev. A}\ }\textbf {\bibinfo {volume}
  {71}},\ \bibinfo {pages} {042319} (\bibinfo {year} {2005})}\BibitemShut
  {NoStop}%
\bibitem [{\citenamefont {Bennett}\ \emph
  {et~al.}(1996{\natexlab{a}})\citenamefont {Bennett}, \citenamefont
  {Brassard}, \citenamefont {Popescu}, \citenamefont {Schumacher},
  \citenamefont {Smolin},\ and\ \citenamefont
  {Wootters}}]{BennettPhysRevLett.76.722}%
  \BibitemOpen
  \bibfield  {author} {\bibinfo {author} {\bibfnamefont {C.~H.}\ \bibnamefont
  {Bennett}}, \bibinfo {author} {\bibfnamefont {G.}~\bibnamefont {Brassard}},
  \bibinfo {author} {\bibfnamefont {S.}~\bibnamefont {Popescu}}, \bibinfo
  {author} {\bibfnamefont {B.}~\bibnamefont {Schumacher}}, \bibinfo {author}
  {\bibfnamefont {J.~A.}\ \bibnamefont {Smolin}},\ and\ \bibinfo {author}
  {\bibfnamefont {W.~K.}\ \bibnamefont {Wootters}},\ }\bibfield  {title}
  {\bibinfo {title} {{Purification of Noisy Entanglement and Faithful
  Teleportation via Noisy Channels}},\ }\href
  {https://doi.org/10.1103/PhysRevLett.76.722} {\bibfield  {journal} {\bibinfo
  {journal} {Phys. Rev. Lett.}\ }\textbf {\bibinfo {volume} {76}},\ \bibinfo
  {pages} {722} (\bibinfo {year} {1996}{\natexlab{a}})}\BibitemShut {NoStop}%
\bibitem [{\citenamefont {Bennett}\ \emph
  {et~al.}(1996{\natexlab{b}})\citenamefont {Bennett}, \citenamefont
  {Bernstein}, \citenamefont {Popescu},\ and\ \citenamefont
  {Schumacher}}]{BennettPhysRevA.53.2046}%
  \BibitemOpen
  \bibfield  {author} {\bibinfo {author} {\bibfnamefont {C.~H.}\ \bibnamefont
  {Bennett}}, \bibinfo {author} {\bibfnamefont {H.~J.}\ \bibnamefont
  {Bernstein}}, \bibinfo {author} {\bibfnamefont {S.}~\bibnamefont {Popescu}},\
  and\ \bibinfo {author} {\bibfnamefont {B.}~\bibnamefont {Schumacher}},\
  }\bibfield  {title} {\bibinfo {title} {Concentrating partial entanglement by
  local operations},\ }\href {https://doi.org/10.1103/PhysRevA.53.2046}
  {\bibfield  {journal} {\bibinfo  {journal} {Phys. Rev. A}\ }\textbf {\bibinfo
  {volume} {53}},\ \bibinfo {pages} {2046} (\bibinfo {year}
  {1996}{\natexlab{b}})}\BibitemShut {NoStop}%
\bibitem [{\citenamefont {Plenio}\ and\ \citenamefont
  {Virmani}(2007)}]{PlenioMeasures}%
  \BibitemOpen
  \bibfield  {author} {\bibinfo {author} {\bibfnamefont {M.~B.}\ \bibnamefont
  {Plenio}}\ and\ \bibinfo {author} {\bibfnamefont {S.}~\bibnamefont
  {Virmani}},\ }\bibfield  {title} {\bibinfo {title} {An introduction to
  entanglement measures},\ }\href {https://arxiv.org/abs/quant-ph/0504163}
  {\bibfield  {journal} {\bibinfo  {journal} {Quant. Inf. Comput.}\ }\textbf
  {\bibinfo {volume} {7}},\ \bibinfo {pages} {1} (\bibinfo {year}
  {2007})}\BibitemShut {NoStop}%
\bibitem [{\citenamefont {Bennett}\ \emph {et~al.}(2000)\citenamefont
  {Bennett}, \citenamefont {Popescu}, \citenamefont {Rohrlich}, \citenamefont
  {Smolin},\ and\ \citenamefont {Thapliyal}}]{PhysRevA.63.012307}%
  \BibitemOpen
  \bibfield  {author} {\bibinfo {author} {\bibfnamefont {C.~H.}\ \bibnamefont
  {Bennett}}, \bibinfo {author} {\bibfnamefont {S.}~\bibnamefont {Popescu}},
  \bibinfo {author} {\bibfnamefont {D.}~\bibnamefont {Rohrlich}}, \bibinfo
  {author} {\bibfnamefont {J.~A.}\ \bibnamefont {Smolin}},\ and\ \bibinfo
  {author} {\bibfnamefont {A.~V.}\ \bibnamefont {Thapliyal}},\ }\bibfield
  {title} {\bibinfo {title} {Exact and asymptotic measures of multipartite
  pure-state entanglement},\ }\href
  {https://doi.org/10.1103/PhysRevA.63.012307} {\bibfield  {journal} {\bibinfo
  {journal} {Phys. Rev. A}\ }\textbf {\bibinfo {volume} {63}},\ \bibinfo
  {pages} {012307} (\bibinfo {year} {2000})}\BibitemShut {NoStop}%
\bibitem [{\citenamefont {Ferrari}\ \emph {et~al.}(2023)\citenamefont
  {Ferrari}, \citenamefont {Lami}, \citenamefont {Theurer},\ and\ \citenamefont
  {Plenio}}]{Ferrari2023}%
  \BibitemOpen
  \bibfield  {author} {\bibinfo {author} {\bibfnamefont {G.}~\bibnamefont
  {Ferrari}}, \bibinfo {author} {\bibfnamefont {L.}~\bibnamefont {Lami}},
  \bibinfo {author} {\bibfnamefont {T.}~\bibnamefont {Theurer}},\ and\ \bibinfo
  {author} {\bibfnamefont {M.~B.}\ \bibnamefont {Plenio}},\ }\bibfield  {title}
  {\bibinfo {title} {{Asymptotic State Transformations of Continuous Variable
  Resources}},\ }\href {https://doi.org/10.1007/s00220-022-04523-6} {\bibfield
  {journal} {\bibinfo  {journal} {Commun. Math. Phys.}\ }\textbf {\bibinfo
  {volume} {398}},\ \bibinfo {pages} {291} (\bibinfo {year}
  {2023})}\BibitemShut {NoStop}%
\bibitem [{\citenamefont {Rains}(2001)}]{959270}%
  \BibitemOpen
  \bibfield  {author} {\bibinfo {author} {\bibfnamefont {E.}~\bibnamefont
  {Rains}},\ }\bibfield  {title} {\bibinfo {title} {A semidefinite program for
  distillable entanglement},\ }\href {https://doi.org/10.1109/18.959270}
  {\bibfield  {journal} {\bibinfo  {journal} {IEEE Trans. Inf. Theory}\
  }\textbf {\bibinfo {volume} {47}},\ \bibinfo {pages} {2921} (\bibinfo {year}
  {2001})}\BibitemShut {NoStop}%
\bibitem [{\citenamefont {Devetak}\ and\ \citenamefont
  {Winter}(2005)}]{Devetak2005}%
  \BibitemOpen
  \bibfield  {author} {\bibinfo {author} {\bibfnamefont {I.}~\bibnamefont
  {Devetak}}\ and\ \bibinfo {author} {\bibfnamefont {A.}~\bibnamefont
  {Winter}},\ }\bibfield  {title} {\bibinfo {title} {Distillation of secret key
  and entanglement from quantum states},\ }\href
  {https://doi.org/10.1098/rspa.2004.1372} {\bibfield  {journal} {\bibinfo
  {journal} {Proc. R. Soc. Lond. A}\ }\textbf {\bibinfo {volume} {461}},\
  \bibinfo {pages} {207} (\bibinfo {year} {2005})}\BibitemShut {NoStop}%
\bibitem [{\citenamefont {Matthews}\ and\ \citenamefont
  {Winter}(2008)}]{Matthews_2008}%
  \BibitemOpen
  \bibfield  {author} {\bibinfo {author} {\bibfnamefont {W.}~\bibnamefont
  {Matthews}}\ and\ \bibinfo {author} {\bibfnamefont {A.}~\bibnamefont
  {Winter}},\ }\bibfield  {title} {\bibinfo {title} {Pure-state transformations
  and catalysis under operations that completely preserve positivity of partial
  transpose},\ }\href {https://doi.org/10.1103/PhysRevA.78.012317} {\bibfield
  {journal} {\bibinfo  {journal} {Phys. Rev. A}\ }\textbf {\bibinfo {volume}
  {78}},\ \bibinfo {pages} {012317} (\bibinfo {year} {2008})}\BibitemShut
  {NoStop}%
\bibitem [{\citenamefont {Datta}\ \emph {et~al.}(2022)\citenamefont {Datta},
  \citenamefont {Kondra}, \citenamefont {Miller},\ and\ \citenamefont
  {Streltsov}}]{datta2022entanglement}%
  \BibitemOpen
  \bibfield  {author} {\bibinfo {author} {\bibfnamefont {C.}~\bibnamefont
  {Datta}}, \bibinfo {author} {\bibfnamefont {T.~V.}\ \bibnamefont {Kondra}},
  \bibinfo {author} {\bibfnamefont {M.}~\bibnamefont {Miller}},\ and\ \bibinfo
  {author} {\bibfnamefont {A.}~\bibnamefont {Streltsov}},\ }\bibfield  {title}
  {\bibinfo {title} {Entanglement catalysis for quantum states and noisy
  channels},\ }\href {https://arxiv.org/abs/2202.05228} {\bibfield  {journal}
  {\bibinfo  {journal} {arXiv:2202.05228}\ } (\bibinfo {year}
  {2022})}\BibitemShut {NoStop}%
\bibitem [{\citenamefont {Seevinck}\ and\ \citenamefont
  {Uffink}(2001)}]{PhysRevA.65.012107}%
  \BibitemOpen
  \bibfield  {author} {\bibinfo {author} {\bibfnamefont {M.}~\bibnamefont
  {Seevinck}}\ and\ \bibinfo {author} {\bibfnamefont {J.}~\bibnamefont
  {Uffink}},\ }\bibfield  {title} {\bibinfo {title} {Sufficient conditions for
  three-particle entanglement and their tests in recent experiments},\ }\href
  {https://doi.org/10.1103/PhysRevA.65.012107} {\bibfield  {journal} {\bibinfo
  {journal} {Phys. Rev. A}\ }\textbf {\bibinfo {volume} {65}},\ \bibinfo
  {pages} {012107} (\bibinfo {year} {2001})}\BibitemShut {NoStop}%
\bibitem [{\citenamefont {G\"uhne}\ and\ \citenamefont
  {T\'oth}(2009)}]{GUHNE20091}%
  \BibitemOpen
  \bibfield  {author} {\bibinfo {author} {\bibfnamefont {O.}~\bibnamefont
  {G\"uhne}}\ and\ \bibinfo {author} {\bibfnamefont {G.}~\bibnamefont
  {T\'oth}},\ }\bibfield  {title} {\bibinfo {title} {Entanglement detection},\
  }\href {https://doi.org/https://doi.org/10.1016/j.physrep.2009.02.004}
  {\bibfield  {journal} {\bibinfo  {journal} {Phys. Rep.}\ }\textbf {\bibinfo
  {volume} {474}},\ \bibinfo {pages} {1} (\bibinfo {year} {2009})}\BibitemShut
  {NoStop}%
\bibitem [{\citenamefont {Palazuelos}\ and\ \citenamefont
  {Vicente}(2022)}]{Palazuelos2022}%
  \BibitemOpen
  \bibfield  {author} {\bibinfo {author} {\bibfnamefont {C.}~\bibnamefont
  {Palazuelos}}\ and\ \bibinfo {author} {\bibfnamefont {J.~I.~d.}\ \bibnamefont
  {Vicente}},\ }\bibfield  {title} {\bibinfo {title} {Genuine multipartite
  entanglement of quantum states in the multiple-copy scenario},\ }\href
  {https://doi.org/10.22331/q-2022-06-13-735} {\bibfield  {journal} {\bibinfo
  {journal} {{Quantum}}\ }\textbf {\bibinfo {volume} {6}},\ \bibinfo {pages}
  {735} (\bibinfo {year} {2022})}\BibitemShut {NoStop}%
\bibitem [{\citenamefont {Piani}\ and\ \citenamefont
  {Mora}(2007)}]{PhysRevA.75.012305}%
  \BibitemOpen
  \bibfield  {author} {\bibinfo {author} {\bibfnamefont {M.}~\bibnamefont
  {Piani}}\ and\ \bibinfo {author} {\bibfnamefont {C.~E.}\ \bibnamefont
  {Mora}},\ }\bibfield  {title} {\bibinfo {title} {Class of
  positive-partial-transpose bound entangled states associated with almost any
  set of pure entangled states},\ }\href
  {https://doi.org/10.1103/PhysRevA.75.012305} {\bibfield  {journal} {\bibinfo
  {journal} {Phys. Rev. A}\ }\textbf {\bibinfo {volume} {75}},\ \bibinfo
  {pages} {012305} (\bibinfo {year} {2007})}\BibitemShut {NoStop}%
\bibitem [{\citenamefont {T\'oth}\ and\ \citenamefont
  {G\"uhne}(2009)}]{PhysRevLett.102.170503}%
  \BibitemOpen
  \bibfield  {author} {\bibinfo {author} {\bibfnamefont {G.}~\bibnamefont
  {T\'oth}}\ and\ \bibinfo {author} {\bibfnamefont {O.}~\bibnamefont
  {G\"uhne}},\ }\bibfield  {title} {\bibinfo {title} {{Entanglement and
  Permutational Symmetry}},\ }\href
  {https://doi.org/10.1103/PhysRevLett.102.170503} {\bibfield  {journal}
  {\bibinfo  {journal} {Phys. Rev. Lett.}\ }\textbf {\bibinfo {volume} {102}},\
  \bibinfo {pages} {170503} (\bibinfo {year} {2009})}\BibitemShut {NoStop}%
\bibitem [{\citenamefont {Christandl}\ and\ \citenamefont
  {Winter}(2004)}]{Christandl_2004}%
  \BibitemOpen
  \bibfield  {author} {\bibinfo {author} {\bibfnamefont {M.}~\bibnamefont
  {Christandl}}\ and\ \bibinfo {author} {\bibfnamefont {A.}~\bibnamefont
  {Winter}},\ }\bibfield  {title} {\bibinfo {title} {{``Squashed
  entanglement'': An additive entanglement measure}},\ }\href
  {https://doi.org/10.1063/1.1643788} {\bibfield  {journal} {\bibinfo
  {journal} {J. Math. Phys.}\ }\textbf {\bibinfo {volume} {45}},\ \bibinfo
  {pages} {829} (\bibinfo {year} {2004})}\BibitemShut {NoStop}%
\bibitem [{\citenamefont {Alicki}\ and\ \citenamefont
  {Fannes}(2004)}]{Alicki_2004}%
  \BibitemOpen
  \bibfield  {author} {\bibinfo {author} {\bibfnamefont {R.}~\bibnamefont
  {Alicki}}\ and\ \bibinfo {author} {\bibfnamefont {M.}~\bibnamefont
  {Fannes}},\ }\bibfield  {title} {\bibinfo {title} {Continuity of quantum
  conditional information},\ }\href
  {https://doi.org/10.1088/0305-4470/37/5/l01} {\bibfield  {journal} {\bibinfo
  {journal} {J. Phys. A}\ }\textbf {\bibinfo {volume} {37}},\ \bibinfo {pages}
  {L55} (\bibinfo {year} {2004})}\BibitemShut {NoStop}%
\bibitem [{\citenamefont {Lami}\ \emph {et~al.}(2024)\citenamefont {Lami},
  \citenamefont {Regula},\ and\ \citenamefont {Streltsov}}]{Lami2023b}%
  \BibitemOpen
  \bibfield  {author} {\bibinfo {author} {\bibfnamefont {L.}~\bibnamefont
  {Lami}}, \bibinfo {author} {\bibfnamefont {B.}~\bibnamefont {Regula}},\ and\
  \bibinfo {author} {\bibfnamefont {A.}~\bibnamefont {Streltsov}},\ }\bibfield
  {title} {\bibinfo {title} {No-go theorem for entanglement distillation using
  catalysis},\ }\href {https://doi.org/10.1103/PhysRevA.109.L050401} {\bibfield
   {journal} {\bibinfo  {journal} {Phys. Rev. A}\ }\textbf {\bibinfo {volume}
  {109}},\ \bibinfo {pages} {L050401} (\bibinfo {year} {2024})}\BibitemShut
  {NoStop}%
\bibitem [{\citenamefont {Chitambar}\ and\ \citenamefont
  {Gour}(2019)}]{RevModPhys.91.025001}%
  \BibitemOpen
  \bibfield  {author} {\bibinfo {author} {\bibfnamefont {E.}~\bibnamefont
  {Chitambar}}\ and\ \bibinfo {author} {\bibfnamefont {G.}~\bibnamefont
  {Gour}},\ }\bibfield  {title} {\bibinfo {title} {Quantum resource theories},\
  }\href {https://doi.org/10.1103/RevModPhys.91.025001} {\bibfield  {journal}
  {\bibinfo  {journal} {Rev. Mod. Phys.}\ }\textbf {\bibinfo {volume} {91}},\
  \bibinfo {pages} {025001} (\bibinfo {year} {2019})}\BibitemShut {NoStop}%
\end{thebibliography}%

\end{document}


\title{Supplemental Material: Catalytic and asymptotic equivalence for quantum entanglement}

\author{Ray Ganardi}
\email{r.ganardi@cent.uw.edu.pl}
\affiliation{Centre for Quantum Optical Technologies, Centre of New Technologies,
University of Warsaw, Banacha 2c, 02-097 Warsaw, Poland}
\author{Tulja Varun Kondra}
\affiliation{Centre for Quantum Optical Technologies, Centre of New Technologies,
University of Warsaw, Banacha 2c, 02-097 Warsaw, Poland}
\author{Alexander Streltsov}
\affiliation{Institute of Fundamental Technological Research, Polish Academy of Sciences, Pawińskiego 5B, 02-106 Warsaw, Poland}

\maketitle

\section{Entanglement theories based on LOCC and PPT operations}

Most of the results discussed in the main text and in this Supplemental Material refer to the framework of entanglement based on LOCC, assuming that Alice and Bob can perform local measurements in their labs, and communicate via a classical communication channel. However, as we will see below, some of the results can be strengthened in a relaxation of LOCC, known as PPT entanglement theory. In this theory, the free states are the states that stay positive under partial transposition (PPT states), and the free operations are the so-called PPT operations, i.e., the set of maps whose Choi state is PPT~\cite{Rains_1999,Rains_1999a}.
It is known that the PPT operations are exactly the operations that preserves the set of PPT states, even when acting only on a subsystem~\cite{Chitambar_2020}.
Moreover, in PPT entanglement theory all non-PPT states are distillable~\cite{Eggeling_2001}.

\section{Proof of Theorem~1}

Theorem~1 of the main text states that marginal reducibility and correlated catalysis are equivalent for bipartite distillable states. The proof of the theorem follows from Propositions~\ref{prop:ReducibilityImpliesCatalysis} and \ref{prop:Distillable} which are given below. In the following, $S$ denotes a possibly multipartite quantum system.

\begin{prop} \label{prop:ReducibilityImpliesCatalysis}
Marginal reducibility from $\rho^S$ onto $\sigma^S$ implies that $\rho^S$ can be converted into $\sigma^S$ via correlated catalysis.
\end{prop}
\begin{proof}

Let $\Lambda$ be an LOCC
protocol converting $n$ copies of an initial state $\rho$
into a state 
\begin{equation}
\Gamma=\Lambda\left[\rho^{\otimes n}\right],\label{eq:Gamma}
\end{equation}
which is a quantum state of the system $S_{1}\otimes S_{2}\otimes\cdots\otimes S_{n}$, and each $S_i$ is a copy of the system $S$.
In the following, $\Gamma_{i}$ denotes the reduced state of $\Gamma$
on $S_{1}\otimes S_{2}\otimes\cdots\otimes S_{i}$ with $\Gamma_{0}=1$.
Moreover, $\Gamma_{i}^{(j)}$ is the reduced state of $\Gamma_{i}$
on $S_{j}$ for $j \leq i$. 

Marginal reducibility of the state $\rho$ into
$\sigma$ implies that for any $\varepsilon>0$ and any $\delta>0$ there
are integers $m \leq n$, and an LOCC protocol $\Lambda$ such that
\begin{subequations}\label{eq:MarginalReducibility}
\begin{align}
\left\Vert \Gamma_{j}^{(j)}-\sigma\right\Vert _{1} & <\varepsilon,\\
\frac{m}{n}+\delta & >1
\end{align}
\end{subequations} for all $j\in[1,m]$.

We are now ready to present a state of the catalyst $\tau$ achieving the transformation $\rho \rightarrow \sigma$: 
\begin{equation}
\tau=\frac{1}{n}\sum_{k=1}^{n}\rho^{\otimes(k-1)}\otimes\Gamma_{n-k}\otimes\ket{k}\!\bra{k},
\end{equation}
in analogy to the construction presented in~\citep{PhysRevLett.127.150503} (see also~\cite{Lipka-Bartosik2102.11846,DuanPhysRevA.71.042319}).
Here, the states $\ket{k}$ are orthonormal states of an auxiliary
system $K$ maintained by Alice. Using the LOCC protocol described
above Eq.~(7) in~\citep{PhysRevLett.127.150503}, the state $\rho^S\otimes\tau^{C}$
is transformed into a state $\mu^{SC}$ with the property that $\mathrm{Tr}_{S}[\mu^{SC}]=\tau^{C}$. Here, $C$ denotes the system of the catalyst.

What remains is to show that $\norm{\mu^S - \sigma^S}_1$ can be made arbitrarily small.
Indeed, we note that the state $\mu^{S}$ can be written as
\begin{equation}
\mu^{S}=\frac{1}{n}\sum_{k=1}^{n}\Gamma_{k}^{(k)}=\frac{1}{n}\sum_{k=1}^{m}\Gamma_{k}^{(k)}+\frac{1}{n}\sum_{k=m+1}^{n}\Gamma_{k}^{(k)}.
\end{equation}
Using Eqs.~(\ref{eq:MarginalReducibility}) we can further write 
\begin{align}
\left\Vert \mu^{S}-\sigma^{S}\right\Vert _{1} & =\frac{1}{n}\left\Vert \sum_{k=1}^{n}\left(\Gamma_{k}^{(k)}-\sigma^{S}\right)\right\Vert _{1}\\
 & \leq\frac{1}{n}\left\Vert \sum_{k=1}^{m}\left(\Gamma_{k}^{(k)}-\sigma^{S}\right)\right\Vert _{1}\nonumber\\
 & +\frac{1}{n}\left\Vert \sum_{k=m+1}^{n}\left(\Gamma_{k}^{(k)}-\sigma^{S}\right)\right\Vert _{1}\nonumber\\
 & <\frac{m}{n}\varepsilon+2\frac{n-m}{n}<\varepsilon+2\delta.\nonumber
\end{align}
The proof is complete by noting that $\varepsilon > 0$ and $\delta > 0$ can be chosen arbitrarily.
\end{proof}

To complete the proof of Theorem~1 of the main text we also need to prove the converse, which is established in the following. We will focus on bipartite systems in the following, i.e., $S=AB$. Extension to multipartite settings will be discussed below.

\begin{prop} \label{prop:Distillable}
    If a distillable state $\rho^{S}$ can be converted into $\sigma^{S}$ via correlated catalysis, then $\rho^{S}$ is reducible onto $\sigma^{S}$ in the marginals.
\end{prop}

\begin{proof}
Let $\tau$ be a state of the catalyst such that 
    \begin{subequations} \label{eq:CorrelatedCatalysis2}
\begin{align}
\mu^{SC} & =\Lambda\left(\rho^{S}\otimes\tau^{C}\right),\\
\left\Vert \mu^{S}-\sigma^{S}\right\Vert _{1} & < \delta,\,\,\,\mu^{C}=\tau^{C}
\end{align}
\end{subequations}
for some $\delta > 0$. 
    We will now show that the conditions for marginal reducibility in Eqs.~(2) of the main text are fulfilled in this case.

    Since the state $\rho$ is distillable, it is possible to distill some singlets and therefore approximate any state $\tau$ via LOCC from a finite number of copies of $\rho$. In more detail, for any $\varepsilon > 0$ there is an integer $k$ and an LOCC protocol $\Lambda'$ such that 
    \begin{equation}
    \left\Vert \Lambda'\left(\rho^{\otimes k}\right)-\tau\right\Vert _{1}<\varepsilon.\end{equation}
    In the following, the state $\tau_\varepsilon = \Lambda'(\rho^{\otimes k})$ will be called \emph{$\varepsilon$-approximation of the catalyst}.
    
    Consider now the following protocol, acting on $\rho^{\otimes n}$:
    \begin{enumerate}
    \item The last $k$ copies of the state $\rho^{\otimes n}$ are converted into $\tau_\varepsilon$ via LOCC, i.e., 
    \begin{equation}
\rho^{\otimes n}\overset{\mathrm{LOCC}}{\longrightarrow}\rho^{\otimes(n-k)}\otimes\tau_{\varepsilon}. \label{eq:TauEpsilon}
\end{equation}
    \item Each of the remaining $n-k$ copies of $\rho$ is converted approximately into the desired state, making repeated use of the state $\tau_\varepsilon$.
    \end{enumerate}

    We will now analyze in more detail the procedure described above. After the first state $\rho^{S_1}$ is converted using the $\varepsilon$-approximation of the catalyst $\tau_\varepsilon^C$, Alice and Bob share the state $\mu_1 \otimes \rho^{\otimes(n-k-1)}$, where the state $\mu_1$ is given as 
    \begin{equation}
    \mu_1^{S_1C}=\Lambda\left(\rho^{S_1}\otimes\tau^{C}_\varepsilon\right),
    \end{equation}
    and $\Lambda$ is the same LOCC protocol as in Eqs.~(\ref{eq:CorrelatedCatalysis2}). Note that
    \begin{equation}
\left\Vert \mu_{1}^{S_{1}C}-\mu^{S_{1}C}\right\Vert _{1}\leq\left\Vert \rho^{S_{1}}\otimes\tau_{\varepsilon}^{C}-\rho^{S_{1}}\otimes\tau^{C}\right\Vert _{1}<\varepsilon.
\end{equation}
Recalling that $\mu^{C} = \tau^{C}$, we obtain the inequalities
\begin{align}
\left\Vert \mu_1^{S_1}-\mu^{S_1}\right\Vert _{1} & <\varepsilon,\\
\left\Vert \mu_1^{C}-\tau^{C}\right\Vert _{1} & <\varepsilon.
\end{align}
Thus, the state $\mu_1^{C}$ is also an $\varepsilon$-approximation of the catalyst state, having the same precision as $\tau_\varepsilon$. 

Alice and Bob now convert the second copy of the state $\rho^{S_2}$, using the catalyst approximation $\mu_1^{C}$. If we define $\mu_{2}^{S_{2}C}=\Lambda(\rho^{S_{2}}\otimes\mu_{1}^{C})$, then by the same arguments as above we will find that 
\begin{align}
\left\Vert \mu_2^{S_2}-\mu^{S_2}\right\Vert _{1} & <\varepsilon,\\
\left\Vert \mu_2^{C}-\tau^{C}\right\Vert _{1} & <\varepsilon.
\end{align}
Iterating this procedure for each copy of $\rho$, we arrive at a state $\nu^{S_1 \ldots S_{n-k}}$ on $S_1 \ldots S_{n-k}$, such that the reduced states on $S_i$ fulfill 
\begin{equation}
\left\Vert \nu^{S_{i}}-\mu^{S_{i}}\right\Vert _{1}<\varepsilon.
\end{equation}
Due to Eqs.~(\ref{eq:CorrelatedCatalysis2}) and the triangle inequality we further find that 
\begin{equation}
\left\Vert \nu^{S_{i}}-\sigma\right\Vert _{1} < \varepsilon + \delta. \label{eq:NuDelta}
\end{equation}

The above arguments show that for every $\varepsilon > 0$ and $\delta > 0$ we can convert $\rho^{\otimes n}$ into a state $\nu^{S_1 \ldots S_{n-k}}$ fulfilling Eq.~(\ref{eq:NuDelta}). Interestingly, while the integer $k$ depends on $\varepsilon$ and $\delta$, the integer $n$ does not depend on these parameters. Thus, we can choose $n$ large enough, making $(n-k)/n$ arbitrarily close to 1. This proves marginal reducibility from $\rho$ to $\sigma$, and the proof of the proposition is complete.
\end{proof}

\section{Marginal and catalytic asymptotic transformation rates}
A key quantity in asymptotic entanglement theory is the transformation rate, describing the optimal performance of an asymptotic transformation of a state $\rho^S$ into another state $\sigma^S$, where $S$ denotes a possibly multipartite systems. We say that an asymptotic transformation from $\rho$ into $\sigma$ is possible at rate $r$ if for any $\varepsilon,\delta > 0$ there are integers $m$, $n$, and an LOCC protocol $\Lambda$ such that
\begin{subequations}
\begin{align}
\left\Vert \Lambda\left(\rho^{\otimes n}\right)-\sigma^{\otimes m}\right\Vert _{1} & <\varepsilon,\\
\frac{m}{n}+\delta & >r.
\end{align}
\end{subequations}
The supremum over such achievable rates $r$ is the asymptotic transformation rate $R(\rho \rightarrow \sigma)$. Rates of this form have been initially studied in the context of singlet distillation~\cite{BennettPhysRevLett.76.722,BennettPhysRevA.53.2046,PlenioMeasures}, which will also be discussed in more detail below. A state $\rho$ is reducible to $\sigma$ in the notion of~\cite{PhysRevA.63.012307} if $R(\rho \rightarrow \sigma) \geq 1$.

Moreover, we say that $\rho$ can be converted into $\sigma$ with correlated catalysis at rate $r$, if for any $\varepsilon > 0$ and any $\delta > 0$ there exist integers $m$, $n$, a catalyst state $\tau$, and an LOCC protocol $\Lambda$ such that
\begin{subequations}
    \begin{align}
\Lambda\left(\rho^{\otimes n}\otimes\tau^{C}\right) & =\mu^{S_{1}\ldots S_{m}C},\\
\left\Vert \mu^{S_{1}\ldots S_{m}}-\sigma^{\otimes m}\right\Vert _{1} & <\varepsilon,\\
\mu^{C} & =\tau^{C},\\
\frac{m}{n}+\delta & >r.
\end{align}
\end{subequations}
The supremum over all such rates will be called \emph{catalytic transformation rate} $R_\mathrm{c}(\rho \rightarrow \sigma)$.

Analogously, we say that a \emph{marginal asymptotic transformation} from $\rho$ to $\sigma$ is possible with rate $r$, if for any $\varepsilon > 0$ and any $\delta > 0$ there exist integers $m$, $n$, and an LOCC protocol $\Lambda$ such that the following equations hold:
\begin{subequations} \label{marginal_rate}
\begin{align} 
\Lambda\left(\rho^{\otimes n}\right)= & \mu^{S_1 \ldots S_m},\label{marginal_rate_1}\\
\left\Vert \mu^{S_i}-\sigma\right\Vert _{1} & <\varepsilon\,\,\,\forall i\leq m,\label{marginal_rate_2}\\
\frac{m}{n}+\delta & > r \label{marginal_rate_3}.
\end{align}
\end{subequations}
Here, $\mu^{S_1 \ldots S_m}$ is a state of the system $S_{1}\otimes S_{2}\otimes\cdots\otimes S_{m}$, where each $S_i$ is a copy of the system $S$. The largest value of $r$ fulfilling these properties will be called \emph{marginal transformation rate} $R_\mathrm{m}(\rho \rightarrow \sigma)$. We note that rate of this form has been defined previously in~\cite{Ferrari2023}. A state $\rho$ is said to be reducible to $\sigma$ in the marginals if $R_\mathrm{m}(\rho \rightarrow \sigma) \geq 1$.

Finally, we say that a state $\rho$ can be converted into $\sigma$ via \emph{marginal asymptotic transformations with correlated catalysis} with rate $r$ if for any $\varepsilon,\delta > 0$ there exist integers $m$, $n$, a state of the catalyst $\tau^C$, and an LOCC protocol $\Lambda$ such that
\begin{subequations} \label{eq:MarginalCatalytic}
\begin{align}
\Lambda\left(\rho^{\otimes n}\otimes\tau^{C}\right) & =\mu^{S_{1}\ldots S_{m}C},\\
\left\Vert \mu^{S_{i}}-\sigma\right\Vert _{1} & <\varepsilon\,\,\,\forall i\leq m,\\
\mu^{C} & =\tau^{C},\\
\frac{m}{n}+\delta & >r.
\end{align}
\end{subequations}
The maximal such rate will be called \emph{marginal catalytic transformation rate} $R_\mathrm{mc}(\rho \rightarrow \sigma)$. It is straightforward to see that 
\begin{align}
R_{\mathrm{mc}}(\rho\rightarrow\sigma) & \geq R_{\mathrm{m}}(\rho\rightarrow\sigma)\geq R(\rho\rightarrow\sigma), \label{eq:RmBound}\\
R_{\mathrm{mc}}(\rho\rightarrow\sigma) & \geq R_{\mathrm{c}}(\rho\rightarrow\sigma)\geq R(\rho\rightarrow\sigma). \label{eq:RcBound}
\end{align}
In analogy to the above definitions, we define $R^{\mathrm{PPT}}$, $R_{\mathrm{m}}^{\mathrm{PPT}}$, $R_{\mathrm{c}}^{\mathrm{PPT}}$, and
$R_{\mathrm{mc}}^{\mathrm{PPT}}$ to be the rates achievable via PPT operations. It is clear that Eqs.~(\ref{eq:RmBound}) and~(\ref{eq:RcBound}) also hold in this setting.

We will now prove that for bipartite distillable states, catalysis cannot enhance marginal asymptotic transformation rates.
\begin{prop} \label{prop:Rates}
For any two bipartite distillable states $\rho$ and $\sigma$ it
holds that 
\begin{equation} 
R_{\mathrm{mc}}\left(\rho\rightarrow\sigma\right)=R_{\mathrm{m}}\left(\rho\rightarrow\sigma\right).
\end{equation}
\end{prop}

\begin{proof}
We will show that a marginal catalytic protocol achieving the rate $R_\mathrm{mc}$ can always be used to construct a marginal protocol without the catalyst, achieving the same rate. For this let $\tau$ be the state of the catalyst such that Eqs.~(\ref{eq:MarginalCatalytic}) are fulfilled. In analogy to the proof of Proposition~\ref{prop:Distillable}, recall that the state $\tau$ can be approximated by a state $\tau_{\varepsilon'}$, which can be obtained via LOCC from a finite number of copies of the initial state $\rho$, i.e., $\tau_{\varepsilon'} = \Lambda'(\rho^{\otimes k})$ and $||\tau_{\varepsilon'} - \tau||_1<\varepsilon'$.

Consider now the following LOCC protocol, acting on $n+k$ copies of $\rho$. In the first step, $k$ copies of the state $\rho$ are converted into $\tau_{\varepsilon'}$ via LOCC. After this step, the total state is given by $\rho^{\otimes n} \otimes \tau_{\varepsilon'}$. In the next step, Alice and Bob apply the LOCC protocol from Eqs.~(\ref{eq:MarginalCatalytic}). The resulting state will be denoted by $\mu_1$, and can be explicitly written as
\begin{equation}
\mu_{1}^{S_{1}\ldots S_{m}C}=\Lambda\left(\rho^{\otimes n}\otimes\tau_{\varepsilon'}^{C}\right).
\end{equation}
Note that 
\begin{equation}
\left\Vert \mu_{1}^{S_{1}\ldots S_{m}C}-\mu^{S_{1}\ldots S_{m}C}\right\Vert _{1}\leq\left\Vert \rho^{\otimes n}\otimes\tau_{\varepsilon'}-\rho^{\otimes n}\otimes\tau_{\varepsilon}\right\Vert _{1}<\varepsilon'
\end{equation}
which implies the inequalities 
\begin{align}
\left\Vert \mu_{1}^{S_{i}}-\mu^{S_{i}}\right\Vert _{1} & <\varepsilon',\\
\left\Vert \mu_{1}^{C}-\tau^{C}\right\Vert _{1} & <\varepsilon'.
\end{align}
The latter inequality implies that $\mu_1^C$ approximates the state $\tau$ with same precision as $\tau_{\varepsilon'}$. Using Eqs.~(\ref{eq:MarginalCatalytic}) and the triangle inequality we further obtain
\begin{equation}
\left\Vert \mu_{1}^{S_{i}}-\sigma\right\Vert _{1}<\varepsilon+\varepsilon'
\end{equation}
for all $i \leq m$.

We will now extend our analysis to $2n + k$ copies of the initial state $\rho$. Again, $k$ copies of $\rho$ will be used to establish the state $\tau_{\varepsilon'}$, resulting in the total state $\rho^{\otimes n}\otimes\rho^{\otimes n}\otimes\tau_{\varepsilon'}$. The first $n$ copies of $\rho$ together with $\tau_{\varepsilon'}$ are converted into the state $\mu_1$, as described above in this proof, leading to the total state $\mu_{1}^{S_{1}\ldots S_{m}C}\otimes\rho^{\otimes n}$. The remaining $n$ copies of $\rho$ are now converted with the LOCC protocol given in Eqs.~(\ref{eq:MarginalCatalytic}), using $\mu_1^C$ as the catalyst state. Recall that $\mu_1^C$ approximates the state $\tau^C$ with the error $\varepsilon'$, which is the same as for $\tau_{\varepsilon'}$. The total state of the systems $S_{m+1} \ldots S_{2m}$ after this transformation will be denoted by $\mu_{2}^{S_{m+1}\ldots S_{2m}C} = \Lambda(\rho^{\otimes n}\otimes \mu_1^C)$. By the same arguments as above, we find that 
\begin{align}
\left\Vert \mu_{2}^{C}-\tau^{C}\right\Vert _{1} & <\varepsilon',\\
\left\Vert \mu_{2}^{S_{i}}-\sigma\right\Vert _{1} & <\varepsilon+\varepsilon'
\end{align}
for all $i \in [m+1,2m]$.

Iterating the above procedure $l$ times, we see that it is possible to convert the state $\rho^{\otimes ln + k}$ into the state $\nu^{S_{1}\ldots S_{lm}}$ having the property that $||\nu^{S_{i}}-\sigma||_{1}<\varepsilon+\varepsilon'$ for all $i \in [1,lm]$. Since this procedure works for any $l$, choosing $l$ large enough we can make $\frac{lm}{ln+k}$ arbitrarily close to $\frac{m}{n}$, and thus also arbitrarily close to $R_\mathrm{mc}$. This proves that it is possible to convert $\rho$ into $\sigma$ with rate $R_\mathrm{mc}$ via marginal asymptotic transformations, and the proof of the proposition is complete.
\end{proof}

We note that in the proof of Proposition~\ref{prop:Rates} we only used the fact that the initial state $\rho$ is distillable. Thus the proposition applies for any target state $\sigma$. Moreover, the proof of Propisition~\ref{prop:Rates} immediately extends to PPT operations, i.e., it holds that 
\begin{equation}
R_{\mathrm{mc}}^{\mathrm{PPT}}(\rho\rightarrow\sigma)=R_{\mathrm{m}}^{\mathrm{PPT}}(\rho\rightarrow\sigma) \label{eq:PPTrates-1}
\end{equation}
whenever the state $\rho$ is distillable under PPT operations. Recalling that all NPT states are distillable via PPT operations~\cite{Eggeling_2001}, it follows that Eq.~(\ref{eq:PPTrates-1}) holds whenever $\rho$ is NPT. We will now show that Eq.~(\ref{eq:PPTrates-1}) also holds if $\rho$ is PPT. Clearly, it is enough to focus on the setting $\rho \in \mathrm{PPT}$ and $\sigma \in \mathrm{NPT}$. In this case we will show that 
\begin{equation}
R_{\mathrm{mc}}^{\mathrm{PPT}}(\rho\rightarrow\sigma)=R_{\mathrm{m}}^{\mathrm{PPT}}(\rho\rightarrow\sigma)=0.
\end{equation}
Since $R_{\mathrm{mc}}^{\mathrm{PPT}} \geq R_{\mathrm{m}}^{\mathrm{PPT}}$, it is enough to prove that $R_{\mathrm{mc}}^{\mathrm{PPT}}(\rho\rightarrow\sigma)=0$ in this setting. For this, let $\mu^{S_{1}\ldots S_{m}C}$ be a quantum state fulfilling Eqs.~(\ref{eq:MarginalCatalytic}) with $\Lambda$ being a PPT operation. Consider the distillable entanglement under PPT operations $E_{\mathrm{d}}^{\mathrm{PPT}}$.  Noting that $E_{\mathrm{d}}^{\mathrm{PPT}}$ is superadditive (see Section~\ref{sec:Superadditivity}), it holds that 
\begin{align}
E_{\mathrm{d}}^{\mathrm{PPT}}(\tau^{C}) & =E_{\mathrm{d}}^{\mathrm{PPT}}(\rho^{\otimes n}\otimes\tau^{C})\geq E_{\mathrm{d}}^{\mathrm{PPT}}(\mu^{S_{1}\ldots S_{m}C})\\
 & \geq E_{\mathrm{d}}^{\mathrm{PPT}}(\mu^{S_{1}\ldots S_{m}})+E_{\mathrm{d}}^{\mathrm{PPT}}(\mu^{C})\nonumber \\
 & =E_{\mathrm{d}}^{\mathrm{PPT}}(\mu^{S_{1}\ldots S_{m}})+E_{\mathrm{d}}^{\mathrm{PPT}}(\tau^{C}),\nonumber 
\end{align}
which implies that the state $\mu^{S_{1}\ldots S_{m}}$ is PPT. Thus, $\mu^{S_i}$ is PPT for all $i \in [1,m]$. Since the target state $\sigma$ is NPT, it has a finite distance from the set of PPT states, which means that in Eqs.~(\ref{eq:MarginalCatalytic}) it is not possible to choose an arbitrarily small $\varepsilon$ whenever the rate $r$ is nonzero. This proves that $R_{\mathrm{mc}}^{\mathrm{PPT}}(\rho\rightarrow\sigma)=0$ in this setting.

We note that the role of catalysis for many-copy transformations between bipartite pure states has been studied earlier in~\cite{DuanPhysRevA.71.042319}. In particular, it was shown that multiple copy transformations, with the aid of a pure catalyst, are equivalent to the single copy catalytic transformation with arbitrary pure catalyst states~\cite{DuanPhysRevA.71.042319}.

\section{Marginal and catalytic transformations for pure target states}

Here we will focus on a bipartite setting with a pure target state $\ket{\phi}$. A case of particular interest is when the target state is a singlet, in which case $R(\rho \rightarrow \psi^-)$ is known as the distillable entanglement $E_\mathrm{d}(\rho)$~\cite{BennettPhysRevLett.76.722,BennettPhysRevA.53.2046,PlenioMeasures}. Using the fact that pure state transformations are asymptotically reversible~\cite{BennettPhysRevA.53.2046}, it is straightforward to see that for pure target states the following holds:
\begin{equation}
R(\rho\rightarrow\phi)=\frac{E_{\mathrm{d}}(\rho)}{S(\phi^{A})}. \label{eq:RatePureTarget}
\end{equation}
We will now show that $R_\mathrm m$ and $R$ coincide for pure target states.

\begin{prop} \label{prop:PureTarget}
The marginal transformation rate coincides with the standard transformation rate for pure target states:
\begin{equation}
R_{\mathrm{m}}(\rho\rightarrow \phi)=R(\rho\rightarrow \phi).
\end{equation}
\end{prop}

\begin{proof}
If $S(\phi^A) = 0$, the target state is not entangled, and both $R(\rho \to \phi)$ and  $R_\mathrm{m}(\rho \to \phi)$ diverge in this case. Without loss of generality we assume $S(\phi^A) > 0$.  

We introduce a slightly different task of transforming a state $\rho$ asymptotically into a state with marginals having distillable entanglement close to $E_\mathrm{d} (\phi) = S(\phi^A)$. Here, we say that a transformation with rate $r$ is possible if for any $\varepsilon,\delta > 0$ there exist integers $m$, $n$ and an LOCC protocol $\Lambda$ such that
\begin{subequations} \label{eq:RmPure}
\begin{align}
\Lambda\left(\rho^{\otimes n}\right) & =\mu^{S_{1}\ldots S_{m}},\\
\left|E_{\mathrm{d}}(\mu^{S_{i}})-S(\phi^{A})\right| & <\varepsilon\,\,\,\forall i\leq m,\\
\frac{m}{n}+\delta & >r.
\end{align}
\end{subequations}
The maximal such rate will be denoted $\tilde{R}_\mathrm{m}(\rho \rightarrow \phi)$.
Recall that the distillable entanglement is bounded as~\cite{959270,Devetak2005}
\begin{equation}
S(\rho^{A})-S(\rho^{AB})\leq E_{\mathrm{d}}(\rho^{AB})\leq S(\rho^{A}).
\end{equation}
This implies that $E_\mathrm d$ is continuous in the vicinity of any pure state, and therefore $\tilde{R}_{\mathrm{m}}(\rho \rightarrow \phi)\geq R_{\mathrm{m}}(\rho \rightarrow \phi)$.

Consider now an LOCC protocol achieving Eqs.~(\ref{eq:RmPure}). Using Eq.~(\ref{eq:RatePureTarget}) and the properties of distillable entanglement (see also Proposition~\ref{prop:EdSuperadditive}) we find
\begin{align}
nR\left(\rho\rightarrow\phi\right) & =n\frac{E_{\mathrm{d}}(\rho)}{S(\phi^{A})}=\frac{E_{\mathrm{d}}\left(\rho^{\otimes n}\right)}{S(\phi^{A})}\geq\frac{E_{\mathrm{d}}\left(\mu^{S_{1}\ldots S_{m}}\right)}{S(\phi^{A})} \label{eq:LOCCRmRequivalence}\\
 & \geq \frac{1}{S(\phi^{A})}\sum_{i=1}^{m}E_{\mathrm{d}}\left(\mu^{S_{i}}\right).\nonumber 
\end{align}
Using this inequality and Eqs.~(\ref{eq:RmPure}) we further obtain
\begin{equation}
nR\left(\rho \rightarrow \phi \right) > m\left(1-\frac{\varepsilon}{S(\phi^A)}\right),
\end{equation}
which implies that 
\begin{equation}
\frac{m}{n}<\frac{R\left(\rho\rightarrow\phi\right)}{1-\frac{\varepsilon}{S(\phi^{A})}}.
\end{equation}
Using Eqs.~(\ref{eq:RmPure}) once again we arrive at 
\begin{equation}
r<\frac{R\left(\rho\rightarrow\phi\right)}{1-\frac{\varepsilon}{S(\phi^{A})}}+\delta.
\end{equation}
Recalling that $\varepsilon,\delta > 0$ can be chosen arbitrarily, we conclude that $\tilde{R}_{\mathrm{m}}(\rho \rightarrow \phi)\leq R(\rho \rightarrow \phi)$. Collecting the above arguments we have 
\begin{equation}
\tilde{R}_{\mathrm{m}}(\rho \rightarrow \phi)\geq R_{\mathrm{m}}(\rho\rightarrow\phi)\geq R(\rho\rightarrow\phi)\geq\tilde{R}_{\mathrm{m}}(\rho \rightarrow \phi),
\end{equation}
which shows that these inequalities are actually equalities. 
\end{proof}

The above proposition also extends to entanglement theory based on PPT operations, i.e., it holds that 
\begin{equation}
R_{\mathrm{m}}^{\mathrm{PPT}}(\rho\rightarrow\phi)=R^{\mathrm{PPT}}(\rho\rightarrow\phi). \label{eq:RmPPT}
\end{equation}
To see this, recall the inequalities
\begin{align}
\ensuremath{S(\rho^{A})-S(\rho^{AB})} & \leq E_{\mathrm{d}}^{\mathrm{LOCC}}(\rho^{AB})\leq E_{\mathrm{d}}^{\mathrm{PPT}}(\rho^{AB})\\
 & \leq E_{\mathrm{c}}^{\mathrm{PPT}}(\rho^{AB})\leq E_{\mathrm{c}}^{\mathrm{LOCC}}(\rho^{AB})\leq S(\rho^{A}),\nonumber 
\end{align}
where $E_{\mathrm{c}}^{X}$ is the entanglement cost under the operation set $X$. For pure target states the rate $R^{\mathrm{PPT}}$ and $E_{\mathrm{d}}^{\mathrm{PPT}}$ are related analogously to Eq.~(\ref{eq:RatePureTarget}) (see also~\cite{Matthews_2008}): 
\begin{equation}
R^{\mathrm{PPT}}(\rho\rightarrow\phi) =\frac{E_{\mathrm{d}}^{\mathrm{PPT}}(\rho)}{S(\phi^{A})}.
\end{equation}

Equipped with these tools, we can define $\tilde{R}_{\mathrm{m}}^{\mathrm{PPT}}$, in analogy to the definition of $\tilde{R}_{\mathrm{m}}$ in Eq.~(\ref{eq:RmPure}), and it holds that
\begin{equation}
    \tilde{R}_{\mathrm{m}}^{\mathrm{PPT}}(\rho\rightarrow\phi) \geq R_{\mathrm{m}}^{\mathrm{PPT}}(\rho\rightarrow\phi).
\end{equation}
In analogy to Eq.~(\ref{eq:LOCCRmRequivalence}) we obtain 
\begin{align}
nR^{\mathrm{PPT}}\left(\rho\rightarrow\phi\right) & =n\frac{E_{\mathrm{d}}^{\mathrm{PPT}}(\rho)}{S(\phi^{A})}=\frac{E_{\mathrm{d}}^{\mathrm{PPT}}\left(\rho^{\otimes n}\right)}{S(\phi^{A})}\geq\frac{E_{\mathrm{d}}^{\mathrm{PPT}}\left(\mu^{S_{1}\ldots S_{m}}\right)}{S(\phi^{A})}\nonumber \\
 & \geq\frac{1}{S(\phi^{A})}\sum_{i=1}^{m}E_{\mathrm{d}}^{\mathrm{PPT}}\left(\mu^{S_{i}}\right).
\end{align}
Using the same arguments as below Eq.~(\ref{eq:LOCCRmRequivalence}) we find that $\tilde{R}_{\mathrm{m}}^{\mathrm{PPT}}(\rho\rightarrow\phi)\leq R^{\mathrm{PPT}}(\rho\rightarrow\phi)$, and thus 
\begin{equation}
\tilde{R}_{\mathrm{m}}^{\mathrm{PPT}}(\rho\rightarrow\phi)\geq R_{\mathrm{m}}^{\mathrm{PPT}}(\rho\rightarrow\phi)\geq R^{\mathrm{PPT}}(\rho\rightarrow\phi)\geq\tilde{R}_{\mathrm{m}}^{\mathrm{PPT}}(\rho\rightarrow\phi).
\end{equation}
This proves that these inequalities are actually equalities.

We note that Proposition~\ref{prop:PureTarget} can be extended also to quantum resource theories different from entanglement, we refer to Proposition~\ref{prop:resource}.

From the above proposition, it follows that in this setting marginal reducibility is equivalent to reducibility as defined in~\cite{PhysRevA.63.012307}. This means that for pure target states in the bipartite setting allowing for correlations between the marginals does not improve the transformation rate. Combining the above results, we can now prove the following proposition.
\begin{prop} \label{prop:RatesPureTarget}
    For any bipartite distillable state $\rho$ and any bipartite pure state $\ket{\phi}$ it holds that 
    \begin{align}
R_{\mathrm{mc}}(\rho\rightarrow\phi) & =R_{\mathrm{m}}(\rho\rightarrow\phi)=R_{\mathrm{c}}(\rho\rightarrow\phi) \label{eq:RatesPureTarget}\\
 & =R(\rho\rightarrow\phi)=\frac{E_{\mathrm{d}}(\rho)}{S(\phi^{A})}.\nonumber 
\end{align}
\end{prop}
\begin{proof}
The proof follows by combining Propositions~\ref{prop:Rates} and~\ref{prop:PureTarget} with Eqs.~(\ref{eq:RcBound}) and (\ref{eq:RatePureTarget}).
\end{proof}

For entanglement theory based on PPT operations we find the following equalities for any bipartite state $\rho$ and any bipartite pure state $\ket{\phi}$:
\begin{align}
R_{\mathrm{mc}}^{\mathrm{PPT}}(\rho\rightarrow\phi) & =R_{\mathrm{m}}^{\mathrm{PPT}}(\rho\rightarrow\phi)=R_{\mathrm{c}}^{\mathrm{PPT}}(\rho\rightarrow\phi) \label{eq:PPTrates}\\
 & =R^{\mathrm{PPT}}(\rho\rightarrow\phi)=\frac{E_{\mathrm{d}}^{\mathrm{PPT}}(\rho)}{S(\phi^{A})}.\nonumber 
\end{align}

Interestingly, it remains unclear whether Proposition~\ref{prop:PureTarget} extends to the multipartite setting, if $\ket{\phi}$ is a general multipartite pure state. However, as we will see below, Proposition~\ref{prop:PureTarget} also applies if $\ket{\phi}$ is a specific multipartite state, comprising a singlet shared between each of the parties.

We will now investigate catalytic transformations with pure target states. We say that a state $\rho$ can be converted into $\sigma$ via \emph{correlated catalysis with decoupling} if for any $\varepsilon > 0$ there is a catalyst state $\tau$ and an LOCC protocol $\Lambda$ such that~\cite{PhysRevLett.127.150503,datta2022entanglement}
\begin{subequations}
    \begin{align}
\left\Vert \Lambda\left(\rho^{S}\otimes\tau^{C}\right)-\sigma^{S}\otimes\tau^{C}\right\Vert _{1} & <\varepsilon,\\
\mathrm{Tr}_{S}\left[\Lambda\left(\rho^{S}\otimes\tau^{C}\right)\right] & =\tau^{C}.
\end{align}
\end{subequations}
Note that in this framework, the correlations between the primary system $S$ and the catalyst $C$ can be made vanishingly small.
As we show in the following proposition, for pure target states $\sigma = \phi$, correlated catalysis is equivalent to correlated catalysis with decoupling. Here, $S$ denotes a possibly multipartite system.

\begin{prop} \label{prop:CatalysisPure}
    A state $\rho^S$ can be converted into a pure state $\ket{\psi}^S$ via correlated catalysis if and only if the conversion is possible via correlated catalysis with decoupling.
\end{prop}
\begin{proof}
Assume that the transformation $\rho \rightarrow \phi$ is possible via correlated catalysis, i.e., for any $\varepsilon>0$ 
there exists a catalyst state $\tau$ and an LOCC protocol $\Lambda$
such that Eqs.~(1) of the main text are fulfilled. Using Lemma~\ref{lem:UsefulLemma} which is given below, we see that the catalyst decouples in this procedure, and moreover 
\begin{equation}
\left\Vert \mu^{SC} - \phi^{S}\otimes\tau^{C}\right\Vert _{1}<\varepsilon+6\sqrt{\frac{\varepsilon}{2}}.
\end{equation}
This shows that the existence of a correlated catalytic transformation from $\rho$ into $\ket{\phi}$ implies that the transformation is also possible via correlated catalysis with decoupling. The converse is straightforward, noting that correlated catalysis is at least as powerful as correlated catalysis with decoupling in general.
\end{proof}
This proposition holds also for the resource theory of entanglement based on PPT operations.

Finally, we will now show that for pure target states in the bipartite setting, correlated catalysis with decoupling is equivalent to the notion of reducibility defined in~\cite{PhysRevA.63.012307}.
\begin{prop}
    The following statements are equivalent for any bipartite distillable state $\rho$ and any bipartite pure state $\ket{\phi}$:
    \begin{enumerate}
        \item $\rho$ is reducible onto $\ket{\phi}$
        \item $\rho$ can be converted into $\ket{\phi}$ via correlated catalysis with decoupling
        \item $E_\mathrm d(\rho) \geq S(\phi^A)$
    \end{enumerate}
\end{prop}
\begin{proof}
    From Theorem~1 of the main text, we see that marginal reducibility from $\rho$ to $\ket{\phi}$ is equivalent to the existence of a correlated catalytic transformation from $\rho$ to $\ket{\phi}$. Proposition~\ref{prop:PureTarget} implies that in this setting marginal reducibility is equivalent to the notion of reducibility defined in~\cite{PhysRevA.63.012307}. Proposition~\ref{prop:CatalysisPure} further implies that correlated catalysis is equivalent to correlated catalysis with decoupling. This proves that conditions 1 and 2 are equivalent. The equivalence of condition 3 follows from Proposition~\ref{prop:RatesPureTarget}.
\end{proof}
The above proposition extends to the resource theory of entanglement based on PPT operations. In this setting, the following statements are equivalent for any bipartite state $\rho$ and any bipartite pure state $\ket{\phi}$:
\begin{enumerate}
        \item $\rho$ is reducible onto $\ket{\phi}$
        \item $\rho$ can be converted into $\ket{\phi}$ via correlated catalysis with decoupling
        \item $E_\mathrm d^\mathrm{PPT}(\rho) \geq S(\phi^A)$
\end{enumerate}
To see this, recall that for PPT entanglement theory Theorem~1 of the main text applies to all states. This means that marginal reducibility from $\rho$ to $\ket{\phi}$ is equivalent to the existence of a correlated catalytic transformation from $\rho$ to $\ket{\phi}$. Moreover, Eq.~(\ref{eq:RmPPT}) implies that in this setting marginal reducibility is equivalent to the notion of reducibility defined in~\cite{PhysRevA.63.012307}. Since Proposition~\ref{prop:CatalysisPure} applies to PPT entanglement theory, correlated catalysis is equivalent to correlated catalysis with decoupling in this case. This shows that conditions 1 and 2 are equivalent. The equivalence of condition 3 follows from Eq.~(\ref{eq:PPTrates}).

We complete this section with the following lemma.
\begin{lem} \label{lem:UsefulLemma}
For any quantum state $\mu^{SC}$ the inequality
    \begin{equation}
\left\Vert \mu^{S}-\ket{\phi}\!\bra{\phi}^{S}\right\Vert _{1}<\varepsilon \label{eq:Decoupling-1}
\end{equation}
implies that 
\begin{equation}
\left\Vert \mu^{SC} - \ket{\phi}\!\bra{\phi}^{S}\otimes\mu^{C}\right\Vert _{1}<\varepsilon+6\sqrt{\frac{\varepsilon}{2}}.\label{eq:DecouplingMain-1}
\end{equation}
\end{lem}

\begin{proof}
Note that Eq.~(\ref{eq:Decoupling-1}) implies the
inequality 
\begin{align}
F\left(\mu^{S},\ket{\phi}\!\bra{\phi}^{S}\right) & \geq\sqrt{1-\frac{1}{2}\left\Vert \mu^{S}-\ket{\phi}\!\bra{\phi}^{S}\right\Vert _{1}}\nonumber \\
 & >\sqrt{1-\frac{\varepsilon}{2}}\label{eq:Decoupling-3}
\end{align}
with fidelity $F(\rho,\sigma)=\mathrm{Tr}\sqrt{\sqrt{\rho}\sigma\sqrt{\rho}}$.
The state $\mu^{S}$ has a purification 
\begin{equation}
\ket{\mu}^{ST}=\sum_{i}\lambda_{i}\ket{i}^{S}\ket{i}^{T}\label{eq:Decoupling-2}
\end{equation}
with the Schmidt coefficient $\lambda_{i}$ sorted in decreasing order.
Due to Eq.~(\ref{eq:Decoupling-3}) we have 
\begin{equation}
\lambda_{0}>\sqrt{1-\frac{\varepsilon}{2}}.\label{eq:Decoupling4}
\end{equation}
Let now $\ket{\nu}^{SCD}$ be a purification of $\mu^{SC}$, and observe
that it can be written as 
\begin{equation}
\ket{\nu}^{SCD}=\sum_{i}\lambda_{i}\ket{i}^{S}\ket{\alpha_{i}}^{CD},
\end{equation}
where $\lambda_{i}$ are the same Schmidt coefficients as in Eq.~(\ref{eq:Decoupling-2})
and $\{\ket{\alpha_{i}}\}$ is an orthonormal basis on $CD$. Noting
that 
\begin{equation}
F\left(\ket{\nu}\!\bra{\nu}^{SCD},\ket{0}\!\bra{0}^{S}\otimes\ket{\alpha_{0}}\!\bra{\alpha_{0}}^{CD}\right)=\lambda_{0},
\end{equation}
and using the fact that the fidelity does not decrease under partial
trace we obtain 
\begin{equation}
F\left(\mu^{SC},\ket{0}\!\bra{0}^{S}\otimes\mathrm{Tr}_{D}\left[\ket{\alpha_{0}}\!\bra{\alpha_{0}}^{CD}\right]\right)>\sqrt{1-\frac{\varepsilon}{2}}.
\end{equation}
Using the inequality $||\rho-\sigma||_{1}/2\leq\sqrt{1-F(\rho,\sigma)^{2}}$
we arrive at 
\begin{equation}
\frac{1}{2}\left\Vert \mu^{SC}-\ket{0}\!\bra{0}^{S}\otimes\mathrm{Tr}_{D}\left[\ket{\alpha_{0}}\!\bra{\alpha_{0}}^{CD}\right]\right\Vert _{1}<\sqrt{\frac{\varepsilon}{2}}.\label{eq:Decoupling-7}
\end{equation}
Noting that the trace norm does not increase under partial trace we obtain 
\begin{align}
\left\Vert \mu^{C}-\mathrm{Tr}_{D}\left[\ket{\alpha_{0}}\!\bra{\alpha_{0}}^{CD}\right]\right\Vert _{1}<2\sqrt{\frac{\varepsilon}{2}}.
\end{align}
We now use the triangle
inequality, arriving at 
\begin{align}
 & \left\Vert \mu^{SC}-\ket{0}\!\bra{0}^{S}\otimes\mu^{C}\right\Vert _{1}\leq\left\Vert \mu^{SC}-\ket{0}\!\bra{0}^{S}\otimes\mathrm{Tr}_{D}\left[\ket{\alpha_{0}}\!\bra{\alpha_{0}}^{CD}\right]\right\Vert _{1}\nonumber \\
 & +\left\Vert \ket{0}\!\bra{0}^{S}\otimes\mathrm{Tr}_{D}\left[\ket{\alpha_{0}}\!\bra{\alpha_{0}}^{CD}\right]-\ket{0}\!\bra{0}^{S}\otimes\mu^{C}\right\Vert _{1}<4\sqrt{\frac{\varepsilon}{2}}.
\end{align}
Using again Eq.~(\ref{eq:Decoupling-7}) we find 
\begin{equation}
\left\Vert \mu^{S}-\ket{0}\!\bra{0}^{S}\right\Vert _{1}<2\sqrt{\frac{\varepsilon}{2}},
\end{equation}
which together with Eq.~(\ref{eq:Decoupling-1}) and triangle inequality
implies that 
\begin{equation}
\left\Vert \ket{\phi}\!\bra{\phi}^{S}-\ket{0}\!\bra{0}^{S}\right\Vert _{1}<\varepsilon+2\sqrt{\frac{\varepsilon}{2}}.
\end{equation}
Using once again the triangle inequality we obtain Eq.~(\ref{eq:DecouplingMain-1}):
\begin{align}
 & \left\Vert \mu^{SC}-\ket{\phi}\!\bra{\phi}^{S}\otimes\mu^{C}\right\Vert _{1}\leq\left\Vert \mu^{SC}-\ket{0}\!\bra{0}^{S}\otimes\mu^{C}\right\Vert _{1}\nonumber \\
 & +\left\Vert \ket{0}\!\bra{0}^{S}\otimes\mu^{C}-\ket{\phi}\!\bra{\phi}^{S}\otimes\mu^{C}\right\Vert _{1}<\varepsilon+6\sqrt{\frac{\varepsilon}{2}}.
\end{align}
\end{proof}

\section{Proof of Theorem~2}

Let us introduce the \emph{catalytic distillable entanglement} 
\begin{equation}
E_{\mathrm{cd}}(\rho)=R_{\mathrm{c}}(\rho\rightarrow\psi^{-}),
\end{equation}
that is, the optimal rate of obtaining singlets with the help of a correlated catalyst. Theorem~2 of the main text states that catalytic distillable entanglement coincides with the standard distillable entanglement for any distillable state, i.e.,
\begin{equation}
E_{\mathrm{d}}(\rho) = E_{\mathrm{cd}}(\rho).
\end{equation}
This is a direct consequence of Proposition~\ref{prop:RatesPureTarget}, recalling that standard distillable entanglement can be written as $E_{\mathrm{d}}(\rho)=R(\rho\rightarrow\psi^{-})$.

A similar statement can be obtained for entanglement theory based on PPT operations. In particular, for any bipartite state $\rho$ it holds that 
\begin{equation}
E_{\mathrm{d}}^{\mathrm{PPT}}(\rho)=E_{\mathrm{cd}}^{\mathrm{PPT}}(\rho),
\end{equation}
where $E_{\mathrm{cd}}^{\mathrm{PPT}}(\rho)=R_{\mathrm{c}}^{\mathrm{PPT}}(\rho\rightarrow\psi^{-})$ is the catalytic distillable entanglement via PPT operations. This follows directly from Eqs.~(\ref{eq:PPTrates}).

\section{Extending Theorems~1 and 2 to multipartite settings}

We can generalize Theorems~1 and 2 of the main text into a multipartite setting as follows:
let us consider distillation into singlets shared between any pair of parties.
Let us denote the parties as $A_1, \ldots, A_n$, and let us call all pairs $\mathbb{P} = \Bqty{\pqty{A_i, A_j} \left| \, i < j \right.}$.
The state that we would like to distill is $\Phi = \bigotimes_{p \in \mathbb{P}} \psi^-_{p}$ (see also Fig.~2 of the main text), and it lives in $\mathcal{H}_{\textrm{total}} = \bigotimes_{p \in \mathbb{P}} \mathcal{H}_p = \bigotimes_{\pqty{A_i, A_j} \in \mathbb{P}} A_i \otimes A_j$.
Note that to each party, we associate a different subsystem for every pair, so the total number of subsystems is $n(n-1)/2$.
Let us call the optimal standard asymptotic rate $E_\mathrm{d}(\rho) = R(\rho \to \Phi)$, and states for which $E_\mathrm{d}\pqty{\rho} > 0$ distillable.
Theorem~1 of the main text follows from Proposition~\ref{prop:ReducibilityImpliesCatalysis} and \ref{prop:Distillable}.
Since Proposition~\ref{prop:ReducibilityImpliesCatalysis} does not assume any bipartite structure, it also holds in the multipartite setting.
Furthermore, for any state that is distillable, the construction provided in the proof of Proposition~\ref{prop:Distillable} also works, so in combination we have Theorem~1 of the main text.

To obtain Theorem~2 of the main text, we relied on Proposition~\ref{prop:Rates} and \ref{prop:PureTarget}.
It is clear that Proposition~\ref{prop:Rates} extends easily to this setting, so we only have to show that Proposition~\ref{prop:PureTarget} also holds, at least around $\Phi$.
In a multipartite setting, $E_\mathrm{d}$ is still monotonic, additive under tensor products and superadditive.
Therefore, it is enough to show that it is lower semi-continuous near $\Phi$.
Now, suppose we have an LOCC map $T$ that maps a state $\rho$ into $\mathcal{H}_{\textrm{total}}$.
Let $r = \min_{p \in \mathbb{P}} E_\mathrm{d}^p (T(\rho))$, where $E_\mathrm{d}^p$ is the bipartite distillable entanglement between the factors in $p$.
Then, for any pair of parties we can obtain at least $r$ singlets per copy of $\rho$.
Furthermore, these distillation protocols can be run independently because they are acting on different factors of $\mathcal{H}_{\textrm{total}}$.
Therefore, we have $E_\mathrm{d} \pqty{\rho} \geq \min_{p \in \mathbb{P}} E_\mathrm{d}^p (T(\rho))$.
Since $E_\mathrm{d}^p$ is simply bipartite distillable entanglement, the hashing bound~\cite{Devetak2005} provides a continuous lower bound.
Therefore we would be finished if we can show that it is tight for $\Phi$.
Now, $\Phi$ contains a singlet for every pair of parties, so we can choose $T$ that separates these singlets for each party.
We can verify that in this case, the hashing bound gives $E_\mathrm{d} \pqty{\Phi} \geq 1$, which is indeed tight.
Therefore, Theorem~2 of the main text also holds in a multipartite setting.

It is reasonable to ask whether the notion of multipartite states considered here is related to the notion of genuine multipartite entanglement. In more detail, a biseparable state is a multipartite state which can be written as a convex combination of product states in some bipartitions. Any state which is not biseparable is called genuinely multipartite entangled~\cite{PhysRevA.65.012107,GUHNE20091,Palazuelos2022}. We note that there exist biseparable states which can be distilled into singlets in any bipartition, see e.g. Eq.~(49) in~\cite{GUHNE20091}. Moreover, there exist genuinely multipartite entangled states which are PPT in all bipartitions, and thus cannot be distilled into pure entangled states~\cite{PhysRevA.75.012305,PhysRevLett.102.170503}. We thus conclude that the notion of genuine multipartite entanglement is different from the notion of being distillable into singlets between each pair of parties.

\section{Superadditivity of asymptotic transformation rates} \label{sec:Superadditivity}

We will show that the transformation rate to go to any pure state $\phi$ is superadditive.
In particular, this means that the bipartite distillable entanglement is superadditive, and the same holds true for the multipartite distillable entanglement as defined above. In the following, $S_1$ and $S_2$ denote two (possibly multipartite) systems.
\begin{prop} \label{prop:EdSuperadditive}
    For any state $\mu^{S_{1} S_{2}}$ and any pure state $\phi$, we have
    \begin{align}
        R\pqty{\mu^{S_{1} S_{2}} \to \phi}
        &\geq
        R\pqty{\mu^{S_{1}} \to \phi} + R\pqty{\mu^{S_2} \to \phi}.
    \end{align}
\end{prop}
\begin{proof}
Let us take a feasible rate $r_i < R\pqty{\mu^{S_i} \to \phi}$ and show that $r_1 + r_2$ is a feasible rate for the transformation $\mu^{S_1 S_2} \to \phi$.
To do so, we have to show that for any $\varepsilon, \delta > 0$, there exist $m, n$, and an LOCC protocol $\Lambda$ such that
\begin{subequations}\label{eq:feasible-rate}
\begin{align}
    \norm{
    \Lambda \pqty{\pqty{\mu^{S_1 S_2}}^{\otimes n}} - \ketbra{\phi}^{\otimes m}
    }_1
    < \varepsilon,
    \\
    \frac{m}{n} + \delta
    > r_1 + r_2.
\end{align}
\end{subequations}

Fix an arbitrary $\varepsilon, \delta > 0$.
Without loss of generality, let us assume $\varepsilon < 1$.
In the following, we denote the space of $S_{i}^{\otimes n}$ as $\mathbb{S}_i$, where $i\in\{1,2\}$.
Note that since $r_i$ are feasible, there exist $m_i$, $n$ and some LOCC protocol $\Lambda_{i}$ such that
\begin{subequations}
\begin{align}
&\left\Vert \Lambda_{i}\left(\left(\mu^{S_i}\right)^{\otimes n}\right)-\ketbra{\phi}^{\otimes m_i}\right\Vert_1 < \frac{\varepsilon^2}{100} < \frac{\varepsilon}{2} ,\label{close}\\
&\frac{m_i}{n}+ \frac{\delta}{2} >r_i.\label{number}
\end{align}
\end{subequations}
Here, $\Lambda_{i}$ is a LOCC map acting on the space $\mathbb{S}_i\equiv S_{i}^{\otimes n_i}$ where $i\in\{1,2\}$.
Note that without loss of generality, we can choose the same $n$ for both systems since otherwise we can take the product $n = n_1 n_2$ and update $m_i, \Lambda_i$ accordingly.
From, Lemma \ref{lem:UsefulLemma} and Eq. (\ref{close}), it follows that
\begin{align*}
    \left\Vert \Lambda_{1}\otimes\openone_{2}\left(\left(\mu^{S_1S_2}\right)^{\otimes n}\right)-\ketbra{\phi}^{\otimes m_1}\otimes \left(\mu^{S_2}\right)^{\otimes n}\right\Vert_1
    &< \frac{\varepsilon^2}{100} + 6 \frac{\varepsilon}{\sqrt{200}}
    \\
    &< \frac{\varepsilon}{2}.
\end{align*}
Using the data processing inequality of trace norm, it follows that
\begin{eqnarray}
\left\Vert \Lambda_{1}\otimes\Lambda_{2}\left(\left(\mu^{S_1S_2}\right)^{\otimes n}\right)-\ketbra{\phi}^{\otimes m_1}\otimes\Lambda_{2}\left(\left(\mu^{S_2}\right)^{\otimes n}\right)\right\Vert_1
< \varepsilon/2 \nonumber.
    \end{eqnarray}
Using the triangle inequality along with Eq. (\ref{close}), we get
\begin{align}
 &\left\Vert\Lambda_{1}\otimes\Lambda_{2}\left(\left(\mu^{S_1S_2}\right)^{\otimes n}\right)-\ketbra{\phi}^{\otimes (m_1+m_2)}\right\Vert_1\leq\nonumber\\ &\left\Vert \Lambda_{1}\otimes\Lambda_{2}\left(\left(\mu^{S_1S_2}\right)^{\otimes n}\right)-\ketbra{\phi}^{\otimes m_1}\otimes\Lambda_{2}\left(\left(\mu^{S_2}\right)^{\otimes n}\right)\right\Vert_1\nonumber \\
 & +\left\Vert\ketbra{\phi}^{\otimes m_1}\otimes\Lambda_{2}\left(\left(\mu^{S_2}\right)^{\otimes n}\right)-\ketbra{\phi}^{\otimes (m_1+m_2)}\right\Vert _{1}\nonumber\\
 &<\varepsilon.\label{epsilon1}
\end{align}
Also note that, from Eq. (\ref{number})
\begin{eqnarray}
    \frac{m_1+m_2}{n}+\delta>r_1 + r_2.\label{delta}
\end{eqnarray}
We see that the Eqs.~(\ref{eq:feasible-rate}) are satisfied by choosing $m = m_1 + m_2$, $\Lambda = \Lambda_1 \otimes \Lambda_2$.
Therefore, $R\pqty{\mu^{S_{1}} \to \phi} + R\pqty{\mu^{S_2} \to \phi}$ is a feasible rate for the transformation $\mu^{S_1 S_2} \to \phi$, and the claim is shown.
\end{proof}
The above proposition holds also for the resource theory of entanglement based on PPT operations, i.e., for any state $\mu^{S_1S_2}$ and any pure state $\ket{\phi}$ we have
\begin{equation}
R^{\mathrm{PPT}}\pqty{\mu^{S_{1}S_{2}}\to\phi}\geq R^{\mathrm{PPT}}\pqty{\mu^{S_{1}}\to\phi}+R^{\mathrm{PPT}}\pqty{\mu^{S_{2}}\to\phi}.
\end{equation}

\section{Properties of catalytic and marginal transformation rates}

In the following we will provide a general upper bound on the catalytic and marginal transformation rates in terms of the squashed entanglement of the corresponding quantum states. Squashed entanglement is defined as~\cite{Christandl_2004}
\begin{equation}
    E_\mathrm{sq}(\rho^{AB}) = \inf \left\{ \frac{1}{2} I(A;B|E):\rho^{ABE} \mathrm{\,\,extension\,\,of\,\,} \rho^{AB} \right\},
\end{equation}
with the quantum conditional mutual information $I(A;B|E) = S(\rho^{AE}) + S(\rho^{BE}) - S(\rho^{ABE})-S(\rho^E)$.

\begin{prop} \label{prop:RatesBound}
The rates $R_{\mathrm{m}}$ and $R_{\mathrm{mc}}$ are bounded
as 
\begin{equation}
R_{\mathrm{m}}\left(\rho\rightarrow\sigma\right)\leq R_{\mathrm{mc}}\left(\rho\rightarrow\sigma\right)\leq\frac{E_{\mathrm{sq}}\left(\rho\right)}{E_{\mathrm{sq}}\left(\sigma\right)}.
\end{equation}
\end{prop}
\begin{proof}
We introduce a slightly different version of the transformation rate, which we will call \emph{squashed transformation rate}. In this framework, we say that a state $\rho$ can be converted into $\sigma$ with rate $r$ if the following inequalities are fulfilled for all $\varepsilon,\delta > 0$:\begin{subequations} \label{eq:SquashedRate}
\begin{align}
\Lambda\left(\rho^{\otimes n}\otimes\tau^{C}\right) & =\mu^{S_{1}\ldots S_{m}C},\\
\left|E_{\mathrm{sq}}(\mu^{S_{i}})-E_{\mathrm{sq}}(\sigma)\right| & <\varepsilon\,\,\,\forall i\leq m,\\
\mu^{C} & =\tau^{C},\\
\frac{m}{n}+\delta & >r.
\end{align}
\end{subequations}
The squashed transformation rate is the maximal such rate, and it will be denoted by $R_\mathrm{sq}$. By continuity of squashed entanglement~\cite{Alicki_2004}, it is clear that $R_{\mathrm{sq}}(\rho\rightarrow\sigma)\geq R_{\mathrm{mc}}(\rho\rightarrow\sigma)$.

Consider now an LOCC protocol $\Lambda$ and a catalyst state $\tau$ achieving Eqs.~(\ref{eq:SquashedRate}). Using the properties of squashed entanglement~\cite{Christandl_2004} we find
\begin{align}
nE_{\mathrm{sq}}\left(\rho\right)+E_{\mathrm{sq}}\left(\tau\right) & =E_{\mathrm{sq}}\left(\rho^{\otimes n}\otimes\tau^{C}\right)\geq E_{\mathrm{sq}}\left(\mu^{S_{1}\ldots S_{m}C}\right)\\
 & \geq\sum_{i=1}^{m}E_{\mathrm{sq}}\left(\mu^{S_{i}}\right)+E_{\mathrm{sq}}\left(\tau\right),\nonumber 
\end{align}
and thus 
\begin{align}
nE_{\mathrm{sq}}\left(\rho\right) & \geq\sum_{i=1}^{m}E_{\mathrm{sq}}\left(\mu^{S_{i}}\right).
\end{align}
Using this inequality and Eqs.~(\ref{eq:SquashedRate}) we further obtain
\begin{equation}
nE_{\mathrm{sq}}\left(\rho\right) > m\left[E_{\mathrm{sq}}\left(\sigma\right)-\varepsilon\right],
\end{equation}
which implies that 
\begin{equation}
\frac{m}{n} < \frac{E_{\mathrm{sq}}\left(\rho\right)}{E_{\mathrm{sq}}\left(\sigma\right)-\varepsilon}.
\end{equation}
Using Eqs.~(\ref{eq:SquashedRate}) once again we arrive at 
\begin{equation}
r<\frac{E_{\mathrm{sq}}\left(\rho\right)}{E_{\mathrm{sq}}\left(\sigma\right)-\varepsilon}+\delta.
\end{equation}
Recalling that $\varepsilon,\delta > 0$ can be chosen arbitrarily, we conclude that
\begin{equation}
R_{\mathrm{sq}}\left(\rho\rightarrow\sigma\right)\leq\frac{E_{\mathrm{sq}}\left(\rho\right)}{E_{\mathrm{sq}}\left(\sigma\right)},
\end{equation}
and the proof is complete.
\end{proof}

We note that the inequality $R_{\mathrm{m}}(\rho\rightarrow\sigma)\leq E_{\mathrm{sq}}(\rho)/E_{\mathrm{sq}}(\sigma)$ has been proven previously in~\cite{Ferrari2023}.

\section{Full equivalence between catalysis and reducibility in PPT entanglement theory}

Now we will show that for the theory of PPT entanglement, correlated catalysis is exactly as powerful as reducibility on the marginals, i.e., Theorem 1 of the main text holds for all bipartite states $\rho$ and $\sigma$. By modifying the arguments that we have presented to use PPT operations instead of LOCC, we have shown that if $\rho$ is distillable, then for any $\sigma$, $\rho$ is reducible to $\sigma$ in the marginals if and only if $\rho$ can be converted to $\sigma$ via correlated catalysis. Recall that in PPT entanglement theory all non-PPT states are distillable~\cite{Eggeling_2001}.
Therefore to show the full equivalence between correlated catalysis and reducibility in the marginals, what remains is to show that when $\rho$ is PPT, then for any $\sigma$, $\rho$ is reducible to $\sigma$ in the marginals if and only if $\rho$ can be converted to $\sigma$ via correlated catalysis.

Let us choose an arbitrary PPT state $\rho$.
We know that $\rho$ is reducible to $\sigma$ in the marginals if and only if $\sigma$ is PPT, because the set of PPT states is closed under taking tensor products, partial trace, and application of PPT operations.
Therefore, we will be done if we show that we cannot transform $\rho$ into a non-PPT state $\sigma$ by correlated catalysis.
Ref.~\cite[Theorem 2]{Lami2023b} fills exactly this gap.
Here, we provide another proof for completeness.
\begin{prop}
Let $\rho$ be a PPT state.
Then $\rho$ can be converted into $\sigma$ with correlated catalysis if and only if $\sigma$ is PPT.
\end{prop}
\begin{proof}
    If $\sigma$ is PPT, there is a PPT operation that transform any state into $\sigma$.
    Therefore, we only need to show that if $\rho$ can be transformed into $\sigma$ with correlated catalysis, then $\sigma$ is PPT.
    Indeed, suppose that $\rho$ can be transformed into $\sigma$ with correlated catalyst.
    We will show that for any $\varepsilon > 0$, there exists a PPT state that is $\varepsilon$-close to $\sigma$.
    Since the set of PPT states are closed, this allows us to conclude that $\sigma$ is also PPT.
    
    Let us fix an $\varepsilon > 0$.
    By definition, there is a catalyst $\tau$ and a PPT operation $\Lambda$ such that $\mu^{SC} = \Lambda{\pqty{\rho^S \otimes \tau^C}}$ satisfies
    \begin{subequations}
    \begin{align}
        \mu^{C} &= \tau,
        \\
        \norm{ \mu^S - \sigma } &\leq \varepsilon.
    \end{align}
    \end{subequations}
    We will show that $\mu^S$ is a PPT state that is $\varepsilon$-close to $\sigma$.
    Recall that the PPT distillable entanglement $E_\mathrm{d}^\mathrm{PPT}(\rho)$ of any state cannot be increased by applying a PPT operation, and that it is superadditive $E_{\mathrm{d}}^\mathrm{PPT}(\rho^{S_1 S_2}) \geq E_{\mathrm{d}}^\mathrm{PPT}(\rho^{S_1}) + E_{\mathrm{d}}^\mathrm{PPT}(\rho^{S_2})$.
    Since $\rho$ is a PPT state, the PPT distillable entanglement of $\tau$ is the same as the PPT distillable entanglement of $\rho \otimes \tau$, because preparing $\rho$ is a PPT operation.
    By using the mentioned properties of PPT distillable entanglement, we have
        \begin{align}
            E_{\mathrm{d}}^\mathrm{PPT} \pqty{\tau^C}
            &= E_{\mathrm{d}}^\mathrm{PPT} \pqty{\rho^S \otimes \tau^C}
            \\
            &\geq E_{\mathrm{d}}^\mathrm{PPT} \pqty{\mu^{SC}} \nonumber
            \\
            &\geq E_{\mathrm{d}}^\mathrm{PPT} \pqty{\mu^S} + E_{\mathrm{d}}^\mathrm{PPT} \pqty{\mu^C} \nonumber
            \\
            &= E_{\mathrm{d}}^\mathrm{PPT} \pqty{\mu^S} + E_{\mathrm{d}}^\mathrm{PPT} \pqty{\tau^C}. \nonumber
        \end{align}
    Therefore, we conclude that $E_{\mathrm{d}}^\mathrm{PPT} \pqty{\mu^S} = 0$.
    Recalling that all non-PPT states are distillable~\cite{Eggeling_2001}, we conclude that $\mu^S$ must be a PPT state.
\end{proof}

\section{General quantum resource theories}

We will now show a generalized version of Proposition~\ref{prop:PureTarget} for general quantum resource theories.
This, under some assumptions, will show an equivalence between asymptotic rates and marginal asymptotic rates in general resource theories.

Any quantum resource theory is defined by a set of free states ($\mathcal{F}$) and a set of free operations ($\mathcal{O}$) \cite{RevModPhys.91.025001}, such that the following property holds
\begin{equation}
    \Lambda_f(\rho_f)\in\mathcal{F}\,\,\forall\,\,\rho_f\in\mathcal{F}\,\,\text{and}\,\,\Lambda_f\in\mathcal{O}.
\end{equation}
With this introduction, we now state our result. As a notation, we denote $R(\rho\rightarrow \sigma)$ as $R_{\sigma}(\rho)$.
\begin{prop} \label{prop:resource}
For any initial state $\rho$ and a target state $\sigma$, 
$R_{\mathrm{m}}(\rho\rightarrow \sigma)=R(\rho\rightarrow \sigma)$ if the following conditions hold
\begin{itemize}
    \item $R_{\sigma}(\mu^{SS'})\geq R_{\sigma}(\mu^{S})+R_{\sigma}(\mu^{S'})$ for all $\mu^{SS'}$ (super-additivity).
\item $R_{\sigma}(\cdot)$ is lower semi-continuous at $\sigma$.
\end{itemize}

\end{prop}
\begin{proof}
Note that, by definition $R_{\mathrm m}(\rho\rightarrow \sigma)\geq R_{\sigma}(\rho)$. Therefore it is enough to show that $R_{\mathrm{m}}(\rho\rightarrow \sigma)\leq R(\rho\rightarrow \sigma)$.

Since, $R_{\sigma}(\cdot)$ is lower semi-continuous at $\sigma$, there exists a small enough $\varepsilon'$, such that for every $\tau$ satisfying $\left\Vert \tau-\sigma\right\Vert _{1} \leq\varepsilon'$, we have
\begin{eqnarray}\label{lower_semi}
    R_{\sigma}(\tau)\geq R_{\sigma}(\sigma)-f(\varepsilon')=1-f(\varepsilon').
\end{eqnarray}
Here, $f(\varepsilon')\rightarrow 0 $ as $\varepsilon'\rightarrow 0^+$. Here, marginal asymptotic transformations can be defined analogous to Eq. (\ref{marginal_rate}), by replacing the LOCC operation $\Lambda$, by a free operation.

Consider now a free operation ($\Lambda$) achieving Eq. (\ref{marginal_rate}). Applying $R_{\sigma}(\cdot)$, on both sides of Eq. (\ref{marginal_rate_1}) we get
\begin{align}
R_{\sigma}\left(\rho^{\otimes n}\right)=nR_{\sigma}\left(\rho\right) & \geq R_{\sigma}\left(\mu^{S_{1}\ldots S_{m}}\right)\\
 & \geq\sum_{i=1}^{m}R_{\sigma}\left(\mu^{S_{i}}\right)\nonumber \\
 & \geq m\left(R_{\sigma}\left(\sigma\right)-f(\varepsilon)\right)\nonumber \\
 & =m\left(1-f(\varepsilon)\right).\nonumber 
\end{align}
Here, we used additivity, monotonicity, super-additivity (assumed) and lower semi-continuity (assumed) of $R_{\sigma}(\cdot)$. This implies 
\begin{equation}
\frac{m}{n}\leq\frac{R_{\sigma}\left(\rho\right)}{1-f(\varepsilon)}.
\end{equation}
Using Eqs.~(\ref{delta}) we arrive at 
\begin{equation}
r<\frac{R_{\sigma}\left(\rho\right)}{1-f(\varepsilon)}+\delta.
\end{equation}
Recalling that $\varepsilon,\delta > 0$ can be chosen arbitrarily small, we conclude that $R_{\mathrm{m}}(\rho\rightarrow \sigma)\leq R(\rho\rightarrow \sigma)$.
This completes the proof. 
\end{proof}

\bibliography{literature}